\pdfoutput=1
\documentclass[12pt,twoside]{amsart}
\usepackage[margin=1.25in,headheight=14pt,headsep=20pt,footskip=30pt]{geometry} %
\usepackage{setspace}
\usepackage{algorithm}
\usepackage{algpseudocode}
\usepackage{amssymb}
\usepackage{amsmath}
\usepackage[hyphens]{url} 
\usepackage[utf8]{inputenc}
\usepackage{mathtools}
\usepackage{amsthm}
\usepackage{color}
\usepackage{nicefrac}
\usepackage{tcolorbox}
\usepackage{natbib}
\usepackage{subcaption}
\usepackage{graphicx}
\usepackage{float}
\usepackage{dsfont}
\usepackage{booktabs}
\usepackage{bm}

\newtheorem{thm}{Theorem}[section]

\newtheorem{prop}{Proposition}[section]
\newtheorem{lem}{Lemma}[section]
\newtheorem{cor}{Corollary}[section]
\newtheorem{ass}{Assumption}
\newtheorem{exm}{Example}[section]
\newtheorem{dfn}{Definition}
\newtheorem{rem}{Remark}

\def\proof{\noindent{\em Proof.}~}
\def\eproof{\mbox{\ }\hfill$\square$}

\def\R{\mathbb{R}}
\def\EE{\mathbb{E}}

\title[Risk Averse Welfare Maximization]{Risk-Averse Welfare Maximization via Marginal Treatment Effects}
\author{Jarrod Burgh\quad\quad Emerson Melo }
\newcommand{\paperkeywords}[1]{\textbf{\textit{Keywords---}} #1}
\thanks{Department of Economics,
	Grinnell College; email: burghjarro@grinnell.edu; Department of Economics, Indiana University Bloomington; e-mail: emelo@iu.edu; }

\begin{document}
	
	\maketitle
	\vspace{18mm}\setcounter{page}{1}
	
	\begin{abstract}
		This paper studies risk-averse treatment allocation when individuals self-select into treatment based on unobserved characteristics. We develop a framework that combines the marginal treatment effect approach to endogenous selection with a general class of coherent risk measures that capture distributional preferences over welfare outcomes. We show that the planner’s problem admits equivalent interpretations in terms of uncertainty aversion, distributional robustness, and worst-case welfare. For law-invariant coherent risk measures, we derive a Kusuoka representation that expresses the planner’s objective as a weighted evaluation of different regions of the welfare distribution and characterize the resulting optimal allocation rule. We further establish finite-sample regret guarantees for empirical risk-averse policy learning, showing how the statistical difficulty of learning a policy depends on the planner’s sensitivity to adverse welfare outcomes. The framework nests the risk-neutral policy learning model of \cite{Kiatagawa_Tetenov_2018} as a special case. An application to the \cite{Card1995} college proximity data demonstrates that incorporating risk aversion can lead to economically meaningful changes in optimal college admission policies.
	\end{abstract}
	\paperkeywords{Social welfare functions, treatment allocation,
		marginal treatment effect, Roy model, coherent risk measures,
		regret, Rademacher complexity.}

	\section{Introduction}\label{s1}
	
	How should a social planner allocate a binary treatment---such as college admission, job training, or a medical intervention---across a heterogeneous population? A large literature beginning with \cite{Manski2004} and \cite{Kiatagawa_Tetenov_2018} studies this problem by selecting a treatment rule that maximizes expected social welfare over a class of feasible policies. This framework has become the foundation of modern policy learning, providing both decision-theoretic foundations for treatment allocation and finite-sample guarantees for empirical policy choice.
	
	Two features of many policy environments remain largely absent from this framework. First, treatment participation is often endogenous: individuals select into treatment based on unobserved characteristics that also affect potential outcomes, so welfare comparisons based on observed outcomes confound treatment effects with selection. Recent work by \citet{Sasaki_Ura_2024} addresses this by combining policy learning with the marginal treatment effect (MTE) framework of \citet{Bjorklund_Moffitt1987} and \citet{Heckman_Vytlacil_2005}. Second, policy-learning methods almost invariably assume the planner is risk neutral, evaluating policies solely through expected welfare---a criterion indifferent to how welfare is distributed and unable to distinguish policies with the same average outcome but different lower-tail welfare.
	
	This paper develops a theory of risk-averse treatment allocation under endogenous selection. We model treatment choice using the MTE framework and represent the planner's preferences by coherent risk measures \citep{Artzner_et_al_1999}. This class includes expected welfare, average value-at-risk (AV@R), entropic risk, mean--semideviation, and many other distribution-sensitive welfare criteria within a common axiomatic framework.
	
	Combining endogenous selection with nonlinear welfare criteria creates identification, optimization, and statistical learning problems absent under either assumption alone: since the welfare distribution generated by a treatment rule depends on the selection mechanism itself, a planner concerned with downside risk must jointly account for treatment effects and the distribution of welfare across latent subpopulations, and since coherent risk measures are nonlinear functionals of welfare distributions rather than pointwise transformations of outcomes, the empirical-process arguments underlying existing policy-learning theory cannot be applied directly, requiring new tools for identification, optimization, and finite-sample analysis.
	
	As in \cite{Manski2004} and \cite{Kiatagawa_Tetenov_2018}, we evaluate policies through regret and study the problem of learning an optimal treatment rule from finite samples. Unlike the existing literature, however, the planner's objective is a general coherent risk measure rather than expected welfare. This seemingly modest change substantially alters both the welfare criterion and the statistical analysis.
	
	\medskip
	\noindent\textbf{Contributions.}
	Our contributions are threefold. First, we develop a general framework for risk-averse treatment allocation under endogenous selection. The planner chooses a treatment rule to maximize a coherent risk measure of the induced welfare distribution, while treatment effects are identified through the MTE framework. The resulting model nests the classical expected-welfare formulation as a special case and accommodates a broad class of distribution-sensitive social welfare objectives within a unified framework.
	
	Second, we characterize the planner's optimization problem. We derive sharp identification bounds for risk-averse welfare, establish robust dual representations, and obtain a Kusuoka representation for law-invariant coherent risk measures. These results provide complementary interpretations of risk-averse treatment allocation. The dual representation shows that coherent risk measures evaluate policies under adverse probability distributions close to a benchmark, while the Kusuoka representation shows that general law-invariant coherent preferences can be expressed as weighted combinations of lower-tail welfare criteria. Together, these results unify robust optimization and distribution-sensitive welfare evaluation within a common treatment-allocation framework.
	
	Third, we establish finite-sample guarantees for empirical risk-averse policy learning. A central technical challenge is that nonlinear risk measures cannot be analyzed using the pointwise contraction arguments employed in existing policy-learning theory. We instead exploit variational representations of coherent risk measures to derive contraction inequalities, symmetrization results, uniform convergence, and regret bounds for a broad class of risk-averse welfare criteria. The resulting regret bounds recover the classical risk-neutral guarantees up to explicit risk-measure-specific complexity constants, thereby identifying a trade-off between protection against adverse welfare outcomes and the statistical complexity of learning an optimal treatment rule.
	
	We illustrate the framework using the college-proximity data of \citet{Card1995}. Holding the identification strategy and policy class fixed, risk-averse objectives produce substantially different treatment allocations from the risk-neutral benchmark, reallocating treatment toward lower-welfare individuals without imposing explicit equity weights or group-specific objectives. In the online Appendix B   we repeat
	the empirical exercise of Section \ref{sec:empirical} using the Job Training Partnership Act (JTPA) experimental data analyzed by \cite{Sasaki_Ura_2024}.  Together the two applications bracket the empirical content
	of the theory, showing not only when a coherent risk measure reorders
	the planner's assignment but also when it provably does not, like in the case of the JTPA data. Finally, the online Appendix D discusses a calibrated simulation  study complementing the empirical analysis by confirming the finite-sample predictions of the theory and illustrating the consequences of ignoring endogenous treatment selection. 
	\subsection{Related Literature}
	
	The literature on policy learning can be organized along two dimensions:
	whether treatment selection is exogenous or endogenous, and whether the
	planner evaluates policies using risk-neutral or risk-averse welfare
	criteria. Table~\ref{tab:literature} summarizes the relationship between
	our paper and the closest existing contributions.
	
	\begin{table}[h]
		\centering
		\small
		\setlength{\tabcolsep}{8pt}
		\begin{tabular}{lp{5.2cm}p{3.2cm}}
			\toprule
			& \textbf{Exogenous selection}
			& \textbf{Endogenous selection}\\
			\midrule
			\textbf{Risk-neutral}
			&
			\cite{Manski2004},
			\cite{Kiatagawa_Tetenov_2018},
			\cite{Mbakop_Tabord-Meehan_2021}
			&
			\cite{Sasaki_Ura_2024},
			\textcolor{black}{\cite{AtheyWager2021}}
			
			\\[6pt]
			
			\textbf{Risk-averse}
			&
			\cite{Qi_et_al_2023},
			\cite{fan2025policylearningalphaexpectedwelfare},
			\cite{CuiHan2025}
			&
			\textbf{This paper}
			\\
			\bottomrule
		\end{tabular}
		
		\caption{\small Position of this paper within the policy-learning literature.}
		
		\label{tab:literature}
		
	\end{table}
	
	Existing work fills the other three cells of Table~\ref{tab:literature}: \citet{Manski2004} and \citet{Kiatagawa_Tetenov_2018} study risk-neutral policy learning under exogenous assignment, extended to endogenous selection via the MTE representation by \citet{Sasaki_Ura_2024} and using IV in \cite{Athey_Wagner_2021}; and \citet{Qi_et_al_2023}, \citet{fan2025policylearningalphaexpectedwelfare}, and \citet{CuiHan2025} study distribution-sensitive (AV@R or quantile-based) objectives under exogenous assignment. Our paper combines these two dimensions.
	
	The paper also draws on coherent risk measures and distributionally robust optimization \citep{Artzner_et_al_1999,Shapiro_Kusuoka2013,FollmerSchied_2025,Shapiro2021be.ch6}, adapting these tools to treatment allocation under endogenous selection, and relates to the minimax treatment-assignment literature \citep{STOYE200970,Stoye2012,Hirano_porter2009,TETENOV2012157}, which addresses robustness under exogenous assignment rather than distribution-sensitivity under endogenous selection.
	
	\subsection{Organization}
	Section~\ref{s2} introduces the causal framework, the MTE-based welfare
	representation, and the risk-neutral benchmark. Section~\ref{s3} develops the
	risk-averse welfare framework: identification bounds, robust duality, and the
	Kusuoka representation. Section~\ref{s4} studies empirical policy learning and
	establishes finite-sample regret guarantees. Section~\ref{sec:empirical}
	presents the empirical application and simulation evidence, and
	Section~\ref{s7} concludes. Proofs are in Appendix~A. The online supplement
	contains the JTPA application (Appendix~B), the full development of the
	mean--semideviation risk measure (Appendix~C), and the full simulation study
	(Appendix~D).

	\section{Treatment Allocation under Risk Aversion}\label{s2}
	
	This section has two objectives. First, we present the basic causal framework and review the standard risk-neutral welfare maximization problem (\cite{Manski2004}, \cite{Kiatagawa_Tetenov_2018}, \cite{AtheyWager2021}, \cite{Mbakop_Tabord-Meehan_2021}, \cite{fan2025policylearningalphaexpectedwelfare}, and \cite{Sasaki_Ura_2024}). Second, we introduce our framework for risk-averse welfare maximization.
	\subsection{Causal Model}\label{s2_causal}
	We consider the following causal Roy model:
	\begin{eqnarray}
		Y &=& DY_1+(1-D)Y_0,\label{Outcome}\\
		D &=& 1\{\tilde{v}(Z)-\tilde{\varepsilon}\geq 0\},
		\label{Selection}
	\end{eqnarray}
	where $Y$ denotes an observed outcome variable, $D$ denotes
	an observed binary treatment variable, $Z$ denotes a vector
	of observed exogenous variables, $Y_0$ and $Y_1$ denote
	unobserved potential outcomes under no treatment and under
	treatment respectively, and $\tilde{\varepsilon}$ denotes an
	unobserved heterogeneity factor of treatment selection. Let
	$\mathcal{Z}$ denote the set of all observable covariates $Z$.
	
	\smallskip
	
	Equation~(\ref{Outcome}) models the outcome production using
	the potential-outcome framework, while
	equation~(\ref{Selection}) models treatment selection via a
	binary threshold-crossing model. The function $\tilde{v}$
	is nonparametric and unknown to the econometrician. The model
	allows for endogeneity in the sense that $(Y_0,Y_1)$ and
	$\tilde{\varepsilon}$ may be statistically dependent even
	conditional on $Z$. To achieve identification, we assume $Z$
	contains excluded exogenous variables (instruments) $Z_0$ as well
	as included exogenous variables $X$. Following
	\cite{Sasaki_Ura_2024}, we adopt the following assumptions.
	
	\begin{ass}\label{Assump_MTE}
		Equations~(\ref{Outcome}) and~(\ref{Selection}) hold, and the
		random vector $Z$ can be written as $(Z_0',X')'$, where:
		\begin{itemize}
			\item[(i)] $\tilde{\varepsilon}$ and $Z_0$ are independent
			given $X$;
			\item[(ii)] $\mathbb{E}[Y_d\mid Z,\tilde{\varepsilon}]=
			\mathbb{E}[Y_d\mid X,\tilde{\varepsilon}]$ and
			$\mathbb{E}[Y_d^2]<\infty$;
			\item[(iii)] $\tilde{\varepsilon}$ is continuously
			distributed with a convex support conditional on $X$.
		\end{itemize}
	\end{ass}
	
	Part~(i) is the only independence assumption imposed on the
	model, allowing arbitrary statistical dependence between
	$(Y_0,Y_1)$ and $\tilde{\varepsilon}$ conditional on $Z$.
	Part~(ii) states the exclusion restriction and bounded second
	moments. Part~(iii) rules out point masses and holes. For
	ease of exposition we define $\mathcal{X}$ as the set of all
	observable covariates $X$.
	
	\smallskip
	
	Following the marginal treatment effect (MTE) literature \citep{Heckman_Vytlacil_2005}, we apply normalizing
	transformations $\varepsilon \triangleq F_{\tilde{\varepsilon}|X}
	(\tilde{\varepsilon})$ and $v(Z)\triangleq F_{\tilde{\varepsilon}|X}
	(\tilde{v}(Z))$ in model~(\ref{Selection}). Under
	Assumption~\ref{Assump_MTE}, $D=1\{v(Z)-\varepsilon\geq 0\}$
	with $\varepsilon|Z\sim\operatorname{Uniform}(0,1)$. We
	therefore work with the normalized model:
	\begin{equation}\label{Selection_UNI}
		D=1\{v(Z)-\varepsilon\geq 0\}\quad\text{with}\quad
		\varepsilon|Z\sim\operatorname{Uniform}(0,1), 
	\end{equation}
	
	where the MTE is defined as:
	$$
	\operatorname{MTE}(x,\bar{\varepsilon})=
	\mathbb{E}[Y_1-Y_0\mid X=x,\varepsilon=\bar{\varepsilon}].
	$$
	
	\subsubsection{Treatment Allocation and Welfare}
	
	The planner chooses a non-randomized policy $\pi : \mathcal{Z} \to \{0,1\}$ from a policy class $\Pi$, a collection of Borel-measurable functions fixed by the planner.\footnote{Leading examples are threshold rules, linear index rules, and decision trees.} For $\pi\in\Pi$, the
	individual outcome is:
	\begin{equation}\label{Outcome_rule}
		Y(\pi(Z))\triangleq\pi(Z)Y_1+(1-\pi(Z))Y_0
		=Y_0+\pi(Z)(Y_1-Y_0).
	\end{equation}
	
	The distribution of $Y(\pi)$ is unidentified due to the
	missing-outcome problem \citep{HECKMAN2007_Handbook_partI}. The policy learning literature has
	therefore focused on \emph{average welfare}:
	$$
	\mathcal{W}(\pi)\triangleq\mathbb{E}[Y(\pi(Z))]\quad\text{for }
	\pi\in\Pi,
	$$
	where the corresponding \emph{conditional average outcome} (or welfare) is denoted as:
	$$
	\mathcal{W}(\pi,z)\triangleq\mathbb{E}[Y(\pi(z))\mid Z=z]
	=\mathbb{E}[Y_0\mid Z=z]+\pi(z)\mathbb{E}[Y_1-Y_0\mid Z=z].
	$$
	
	From previous expressions, it is straightforward to see that the risk neutral welfare $\mathcal{W}(\pi)$ is just the average  value of $\mathcal{W}(\pi, Z)$.

	\cite{Sasaki_Ura_2024}
	show  that under Assumption~\ref{Assump_MTE}, the  \emph{risk neutral} welfare can be represented as:
	\begin{equation}\label{W_representation}
		\mathcal{W}(\pi)=\mathbb{E}[Y_0]+\mathbb{E}\!\left[
		\pi(Z)\int_0^1\operatorname{MTE}(X,\nu)\,d\nu
		\right]\quad\forall\,\pi\in\Pi.
	\end{equation}
	Accordingly, it follows that
	\begin{equation}\label{W_conditional_representation_}
		\mathcal{W}(\pi,Z)=\mathbb{E}[Y_0\mid Z]+\pi(Z)
		\int_0^1\operatorname{MTE}(X,\nu)\,d\nu.
	\end{equation}
	Following \cite{Manski2004} and \cite{Kiatagawa_Tetenov_2018}, the risk-neutral planner's decision rule is evaluated by its \emph{maximum regret}:
	\begin{equation}\label{Regret_Def}
		\mathcal{R}(\pi)\triangleq\sup_{\pi'\in\Pi}\mathcal{W}(\pi')-\mathcal{W}(\pi),
	\end{equation}
	where expression (\ref{Regret_Def}) quantifies the welfare loss, relative to the best feasible rule, from implementing $\pi$ instead of the optimal one. Section~\ref{s4} generalizes~(\ref{Regret_Def}) to the risk-averse criterion~(\ref{Regret_Rho}) below and establishes finite-sample bounds on it.

	Throughout the paper we use the following adaptation of Assumption 2(i) in \cite{Sasaki_Ura_2024}. 
	\begin{ass}\label{Assumption2_OCE_Bounded}
		The distribution of
		$(Y_0,Y_1,D,Z)$ satisfies
		$$\left|\int_0^1\operatorname{MTE}(X,\nu)\,d\nu\right|\leq
		M<\infty \quad \text{a.s.}$$
		\noindent for a known constant $M$
	\end{ass}

	\medskip

	Let
	\begin{equation}\label{f_definition}
		f(Z;\pi)\triangleq \pi(Z)\int_0^1\operatorname{MTE}(X,\nu)\,d\nu,
	\end{equation}
	and  define $\mathcal{F}\triangleq \{f(\cdot;\pi):\textcolor{black}{\pi\in\Pi}\}$.
	Then, by Assumptions \ref{Assump_MTE} and \ref{Assumption2_OCE_Bounded},  $\mathcal{F}$ is a class of uniformly bounded functions
	with $\|f\|_\infty\leq M$ for all $f\in\mathcal{F}$. Furthermore,  by
	construction, $\mathcal{W}(\pi,Z)=\mathcal{W}(f,Z)$.
	Throughout, we use $\mathcal{W}(f)=\mathbb{E}[Y(\pi)]$ and
	$\mathcal{W}(f,Z)=\mathbb{E}[Y(\pi)\mid Z]$ for $f = f(\cdot,\pi)$ interchangeably. We impose one further assumption for tractability.
	
	\begin{ass}\label{Ass_Y0_Bounded}
		$|Y_0| \le B < \infty$ almost surely.
	\end{ass}
	
	Assumptions~\ref{Assump_MTE}, \ref{Assumption2_OCE_Bounded},
	and~\ref{Ass_Y0_Bounded} imply that $\mathcal{W}(f,Z)$ is bounded almost
	surely, uniformly over $f\in\mathcal{F}$, so $\mathcal{W}(f,Z)\in L^\infty$ for
	every $f\in\mathcal{F}$.\footnote{Explicit bounds are given in
		Section~\ref{sec:contraction_symmetrization}.} No bound on $Y_1$ is needed:
	welfare depends on the treated outcome only through the conditional means in
	the MTE, which are restricted by Assumption~\ref{Assump_MTE}(ii) and the
	integrated MTE bound in Assumption~\ref{Assumption2_OCE_Bounded}.
	
	\medskip
	
	The welfare representation~\eqref{W_representation} is silent about the
	distributional consequences of a policy. Yet interventions may harm vulnerable
	subpopulations, which makes risk aversion a natural concern for the planner.
	The next section develops a risk-averse welfare framework that addresses this
	while accounting for participants' self-selection.

	\subsection{Risk-Averse Welfare Maximization}
	
	We can now turn our attention to integrating risk aversion in the treatment allocation problem. To do so, we deploy the notion of coherent
	risk measures \citep{Shapiro2021be.ch6} and its relaxation concave risk measures. \textcolor{black}{Throughout the paper we work with the space
		$\mathcal{Z}\triangleq L_p(\Omega,\mathcal{B},\mathbb{P})$,
		$p\in[1,\infty)$}. Equipped with the usual norm, $\mathcal{Z}$
	is a Banach space with dual
	$\mathcal{Z}^*=L_q(\Omega,\mathcal{B},\mathbb{P})$,
	$1/p+1/q=1$.
	
	\begin{dfn}\label{Coherent_Def}
		A map $\rho:\mathcal{Z}\to\bar{\mathbb{R}}$ is a
		\textbf{coherent risk measure} if it satisfies:
		\begin{itemize}
			\item[(A1)] \textit{Monotonicity:} $Z_1\preceq Z_2$
			implies $\rho(Z_1)\leq\rho(Z_2)$.
			\item[(A2)] \textit{Translation invariance:}
			$\rho(Z+a)=\rho(Z)+a$ for all $a\in\mathbb{R}$.
			\item[(A3)] \textit{Concavity:}
			$\rho(tZ_1+(1-t)Z_2)\geq t\rho(Z_1)+(1-t)\rho(Z_2)$
			for all $Z_1,Z_2\in\mathcal{Z}$, $t\in[0,1]$.
			\item[(A4)] \textit{Positive homogeneity:}
			$\rho(\lambda Z)=\lambda\rho(Z)$ for all $\lambda\geq 0$.
		\end{itemize}
		Likewise, a map $\rho$ is a
		\textbf{concave risk measure} if it satisfies (A1)-(A3).
	\end{dfn}
	
	Coherent and concave risk measures are widely used to quantify uncertainty in portfolio choice, certainty-equivalent representations, and scoring rules (\cite{FollmerSchied_2025}). We extend this framework to treatment allocation and welfare maximization, noting the importance of each property. Monotonicity implies that a policy yielding weakly higher welfare is weakly preferred. Translation invariance ensures that welfare comparisons are independent of the baseline level. Concavity captures a preference for diversification across policies, whereas positive homogeneity requires the welfare criterion to scale proportionally with outcomes. Although the scale invariance implied by (A4) may be restrictive in some allocation settings, our main results hold for the broader class of concave risk measures. We also introduce the additional property of law-invariance \citep{Shapiro_et_al_2013} which implies a risk measure is purely driven by the distribution of welfare, not particular subgroups' welfare.\footnote{In the welfare context, law invariance is a form of anonymity: the planner ranks welfare distributions without regard to the identities of those who experience each level. This is the standard impartiality requirement of social welfare analysis; absent it, $\rho$ could encode preferences over \emph{who} receives welfare rather than over its distribution.}
	
	\begin{dfn}
		The risk measure $\rho:\mathcal{Z}\to\mathbb{R}$ is \emph{law invariant} if $Z_1$ and $Z_2$ having the same distribution under $\mathbb{P}$ implies $\rho(Z_1)=\rho(Z_2).$
	\end{dfn}
	\smallskip
	
	Endowed with our general notion of risk measure $\rho$, and recalling  $\mathcal{W}(f,Z)\triangleq\mathbb{E}[Y(\pi)\mid Z]$, a risk-averse social planner solves the following optimization problem:
	\begin{equation}\label{Risk_Averse_program}
		\sup_{f\in\mathcal{F}}\rho(\mathcal{W}(f,Z)).
	\end{equation}

	Accordingly, we define the regret associated with a policy $f\in \mathcal{F}$ as 
	\begin{equation}\label{Regret_Rho}
		\mathcal{R}_\rho(f)=\sup_{f^\prime\in \mathcal{F}}\rho(\mathcal{W}(f^\prime,Z))-\rho(\mathcal{W}(f,Z)).
	\end{equation}
	
	We conclude by highlighting two features of \eqref{Risk_Averse_program} and \eqref{Regret_Rho}. First, the formulation in \eqref{Risk_Averse_program} accommodates a broad class of social-planner risk preferences through the choice of the risk functional $\rho$. In particular, when $\rho=\EE$, it reduces to the standard risk-neutral welfare maximization problem. Second, \eqref{Regret_Rho} generalizes the risk-neutral regret criterion in \eqref{Regret_Def}. In both formulations, the choice of $\rho$ governs the properties of the resulting welfare and regret criteria.
	\smallskip
	
	The following examples illustrate three widely used risk measures. 
	\begin{exm}\label{Example1}
		The (lower) Average Value-at-Risk: For a given $f\in\mathcal{F}$, the Average-value-at-risk at the level $\beta$  associated with   $\mathcal{W}(f,Z)$ is given by the \cite{Rockafellar2000OptimizationOC}'s variational representation
		$$
		\operatorname{AV@R}_\beta(\mathcal{W}(f,Z))\triangleq
		\sup_{\lambda\in\mathbb{R}}\!\left\{
		\lambda+\tfrac{1}{\beta}\mathbb{E}[\mathcal{W}(f,Z)-\lambda]_-
		\right\},\quad
		\beta\in(0,1],
		$$
		where $[t]_-\triangleq\min\{t,0\}$.  When  $\mathcal{W}(f,Z)$ is  a continuous random variable  $\operatorname{AV@R}_\beta(\mathcal{W}(f,Z))$ reduces to:
		$$
		\operatorname{AV@R}_\beta(\mathcal{W}(f,Z))
		=\frac{1}{\beta}\int_0^\beta F_{\mathcal{W}(f,Z)}^{-1}(\tau)\,d\tau,\quad
		\beta\in(0,1],
		$$
		
		In the case where $\beta=1$, we get $\mathrm{AV@R}_1(\mathcal{W}(f,Z))=\mathbb{E}[\mathcal{W}(f,Z)]$, recovering
		the risk-neutral case. Notably, $\mathrm{AV@R}_\beta$ satisfies conditions (A1)–(A4) in Definition \ref{Coherent_Def}, and is therefore a coherent risk measure for all $\beta \in(0,1]$.
	\end{exm}
	
	\begin{exm}\label{Example2} The mean--semideviation risk: For $f\in \mathcal{F}$ the mean--semideviation of order $p\geq 1$ is defined as:
		$$\rho_{p}(\mathcal{W}(f,Z))\triangleq\EE[\mathcal{W}(f,Z)]-c\mathbb{D}_p^-
		[\mathcal{W}(f,Z)],$$ where $\mathbb{D}_p^-[\mathcal{W}(f,Z)]\triangleq(\mathbb{E}[(\mathbb{E}[\mathcal{W}(f,Z)]-\mathcal{W}(f,Z))_+^p])^{1/p}$, $(t)_+\triangleq\max\{t,0\}$, and $c\in[0,1]$. When $c=0$ this reduces to the risk-neutral case for any $p\geq 1$, and the online Appendix C discusses in detail the properties of this risk measure. 
	\end{exm}
	
	\begin{exm}\label{Example3}
		The entropic risk: For $f\in \mathcal{F}$ the entropic risk measure is defined as:
		$$\rho_{ent}(\mathcal{W}(f,Z))=-\frac{1}{\gamma}\ln\mathbb{E}[e^{-\mathcal{W}(f,Z)\gamma}],$$
		$\gamma>0$. This satisfies (A1), (A2), and (A3) but notably not (A4) and is therefore concave. It reflects
		constant absolute risk aversion (CARA) preferences for the planner.
	\end{exm}
	
	Our results apply to general coherent and concave risk measures and therefore extend beyond the preceding examples. Some widely used risk criteria, however, do not satisfy the defining properties of these classes. Although $\operatorname{AV@R}$ is coherent, value at risk ($\operatorname{V@R}$) generally fails to satisfy concavity (A3) and is therefore neither concave nor coherent.\footnote{ This is the technical reason to avoid quantiles as risk measures} Likewise, the certainty-equivalent induced by constant-relative-risk-aversion (CRRA) preferences fails translation invariance (A2) and hence cannot be represented by either a coherent or a concave risk measure.
	
	\subsection{Risk-Averse Welfare Bounds}\label{sec:risk-bounds}
	
	For each $f\in\mathcal{F}$, the distribution of the potential outcome $Y(\pi)$ is generally not identified because of the missing-outcome problem \citep{HECKMAN2007_Handbook_partI}. The following result establishes that the identified conditional average outcome, $\mathcal{W}(f,Z)$, delivers sharp bounds on any concave risk-averse welfare functional of $Y(\pi)$.
	\begin{thm}\label{Risk_Bounds}
		Let Assumptions~\ref{Assump_MTE}, \ref{Assumption2_OCE_Bounded}, and~\ref{Ass_Y0_Bounded} hold. In addition, assume
		that $\mathcal{W}(f,Z)-Y(\pi)\leq K_0$ for some $K_0\in\mathbb{R}$.
		Then for any concave risk functional $\rho$ satisfying (A1)-(A3):
		\begin{equation}\label{Bounds_Welfare}
			\sup_{f\in\mathcal{F}}\rho(\mathcal{W}(f,Z))-K_0
			\;\leq\;
			\sup_{f\in\mathcal{F}}\rho(Y(\pi))
			\;\leq\;
			\sup_{f\in\mathcal{F}}\rho(\mathcal{W}(f,Z)).
		\end{equation}
	\end{thm}
	
	\proof All proofs are gathered in Appendix~A. \eproof
	
	\smallskip
	
	Theorem~\ref{Risk_Bounds} establishes that the identified optimization problem
	$$
	\sup_{f\in\mathcal{F}} \rho\bigl(\mathcal{W}(f,Z)\bigr)$$
	is informative about the unidentified individual-level welfare criterion $\rho(Y(\pi))$. For any risk functional $\rho$ satisfying (A1), (A2), and (A3), the theorem provides computable sharp bounds based solely on the identified object $\mathcal{W}(f,Z)$, providing a formal justification for studying \eqref{Risk_Averse_program}.
	
	To illustrate the implications of Theorem~\ref{Risk_Bounds}, we now specialize to the case $\rho=\operatorname{AV@R}_\beta$.
	\begin{prop}\label{CVAR_cor}
		Let Assumptions~\ref{Assump_MTE}, \ref{Assumption2_OCE_Bounded}, and~\ref{Ass_Y0_Bounded} hold, and suppose that
		$$
		\mathcal{W}(f,Z)-Y(\pi)\leq K_0
		$$
		for some $K_0\in\mathbb{R}$. Then, for any $\beta\in(0,1]$,
		\begin{equation}\label{AVaR_Bounds}
			\sup_{f\in\mathcal{F}}
			\mathrm{AV@R}\beta(\mathcal{W}(f,Z))-K_0
			\leq
			\sup_{f\in\mathcal{F}}\mathrm{AV@R}_\beta(Y(\pi))
			\leq
			\sup_{f\in\mathcal{F}}
			\mathrm{AV@R}_\beta(\mathcal{W}(f,Z)).
		\end{equation}
		Moreover, as $\beta\to1$, the bounds in \eqref{AVaR_Bounds} converge to the risk-neutral bounds implied by \cite{Sasaki_Ura_2024}. As $\beta\to0$, they converge to bounds based on the worst-case welfare criterion.
	\end{prop}
	
	\begin{rem}\label{Rem_K}
		Under Assumption~\ref{Assump_MTE}, suppose additionally that
		$|Y_0|\leq B<\infty$ (Assumption \ref{Ass_Y0_Bounded}) and $|Y_1|\leq B<\infty$ almost surely. Then one may take the explicit (albeit conservative) bound $ K_0=2B+M.$
		Indeed, since $\|f\|_\infty\leq M$ (by the definition of $\mathcal{F}$ following~\eqref{f_definition}) and
		$\EE[Y(\pi)\mid Z]=\EE[Y_0\mid Z]+f(Z;\pi)$ (for all $\pi$)
		Jensen’s inequality implies
		$$
		|\mathcal{W}(f,Z)|
		=|\EE[Y(\pi)\mid Z]|
		\leq |\EE[Y_0\mid Z]|+|f(Z;\pi)|
		\leq B+M.
		$$
		Moreover, because
		$
		Y(\pi)=(1-\pi(Z))Y_0+\pi(Z)Y_1$
		is a convex combination of $Y_0$ and $Y_1$, we have
		$
		|Y(\pi)|\leq B$ a.s. Hence,
		$$
		\mathcal{W}(f,Z)-Y(\pi)
		\leq
		|\mathcal{W}(f,Z)|+|Y(\pi)|
		\leq
		2B+M,
		$$
		uniformly over $f\in\mathcal{F}$, establishing the claim.
		
		The bound is conservative because it relies only on the uniform bounds on the potential outcomes. Sharper bounds may be obtained by exploiting additional distributional information, such as conditional second moments or conditional variances. Finally, note that the constant $K_0$ is specific to the identification bounds in Theorem~\ref{Risk_Bounds} and Proposition~\ref{CVAR_cor}, and should not be confused with the constant $K$ introduced in Section~\ref{s4}.
	\end{rem}
	
	\section{Robustness and Distributional Welfare Analysis }\label{s3}
	
	This section develops the robustness and distributional implications of coherent and concave risk measures for welfare maximization. We first show that convex duality provides a distributionally robust interpretation in which allocation rules are evaluated under worst-case probability distributions in a neighborhood of a benchmark distribution. We then study $\varphi$-divergence risk measures, which admit tractable robust allocation problems, and derive a Kusuoka representation that characterizes the planner’s preferences over the lower tail of the welfare distribution.
	\subsection{Dual Representation and Robust Interpretation}\label{sec:dual_rep}
	
	A central implication of modeling risk-averse treatment allocation with coherent and concave risk measures is its connection to robustness \citep{Hansen_Sargent_2001,Hansen_Sargent2008} and uncertainty aversion \citep{Marinaccietal2006}. The following proposition establishes that, under general conditions, the planner’s risk-averse welfare maximization problem admits an equivalent robust representation.
	\begin{prop}\label{Risk_measure_robust}
		Let Assumptions~\ref{Assump_MTE}, \ref{Assumption2_OCE_Bounded}, and~\ref{Ass_Y0_Bounded} hold. Suppose that
		$\rho:\mathcal{Z}\to\mathbb{R}$ is a proper, concave, and upper semicontinuous risk measure. Then
		$$
		\sup_{f\in\mathcal{F}}\rho\bigl(\mathcal{W}(f,Z)\bigr)
		=
		\sup_{f\in\mathcal{F}}
		\inf_{\xi\in\mathcal{D}^*}
		\left\{
		\langle\xi,\mathcal{W}(f,Z)\rangle
		-
		\rho^*(\xi)
		\right\},
		$$
		where $\rho^*$ denotes the concave conjugate of $\rho$ and
		$
		\mathcal{D}^* \triangleq \operatorname{dom}(\rho^*).$
		
	\end{prop}%
	
	Two observations are worth noting. First, $-\rho^*$ serves as an ambiguity index in the sense of \cite{Marinaccietal2006}. Once the risk measure $\rho$ is specified, the corresponding ambiguity index is determined by its concave conjugate. Second, the assumptions imposed on $\rho$ are mild, so the robust and ambiguity-averse interpretation applies to a broad class of risk measures.
	
	Under the additional assumption that $\rho$ is coherent, the next result provides an explicit characterization of the associated ambiguity set.
	
	\begin{prop}\label{Prop_Coherent_Dual}
		Let Assumptions~\ref{Assump_MTE}, \ref{Assumption2_OCE_Bounded}, and~\ref{Ass_Y0_Bounded} hold. Suppose that
		$\rho:\mathcal{Z}\to\mathbb{R}$ is coherent. Define
		\begin{equation}\label{Dual_Set_D}
			\mathcal{D}\triangleq
			\left\{
			\xi\in\mathcal{Z}^*:
			\langle\xi,Z\rangle\leq\rho(Z)\ \forall\,Z\in\mathcal{Z},\;
			\langle\xi,\mathbf{1}\rangle=1,\;
			\xi\geq0
			\right\}.
		\end{equation}
		Then $\mathcal{D}$ is a nonempty, convex, weak$^*$-compact set of probability densities, and
		\begin{equation}\label{Worst_Case_Expectation}
			\sup_{f\in\mathcal{F}}
			\rho\bigl(\mathcal{W}(f,Z)\bigr)
			=
			\sup_{f\in\mathcal{F}}
			\inf_{\xi\in\mathcal{D}}
			\mathbb{E}_{\xi}\!\left[\mathcal{W}(f,Z)\right].
		\end{equation}
		Moreover, if $\rho=\mathrm{AV@R}_{\beta}$, then
		$
		\mathcal{D}
		=
		\left\{
		\xi\in\mathcal{Z}^*:
		0\leq\xi\leq\frac{1}{\beta}\ \mathbb{P}\text{-a.s.},
		\;
		\mathbb{E}[\xi]=1
		\right\}.
		$
	\end{prop}
	
	The set $\mathcal{D}$ represents the collection of probability measures consistent with the planner’s risk preferences. In the spirit of \cite{GILBOA1989141} and \cite{Marinaccietal2006}, risk aversion can therefore be interpreted as aversion to Knightian uncertainty: the planner evaluates welfare under the least favorable distribution in $\mathcal{D}$.
	
	Propositions~\ref{Risk_measure_robust} and~\ref{Prop_Coherent_Dual} adapt classical duality results from the risk-measure literature \citep{Shapiro_et_al_2013} to the problem of risk-averse welfare maximization. To the best of our knowledge, this robust representation has not previously been studied in the policy-learning or treatment-allocation literature.
	\subsection{Robust Treatment Allocation via $\varphi$-Divergences}\label{oce}
	
	The preceding analysis characterizes the robust optimization problem at a high level of generality. We now specialize to risk measures induced by $\varphi$-divergences. This class plays a central role in the literature on robust control in macroeconomics \citep{Hansen_Sargent_2001,Hansen_Sargent2008,Strzalecki2011} and has also become a standard tool in distributionally robust optimization \citep{Duchietal2021,Kuhn_Shafiee_Wiesemann_2025}. Within the general framework developed above, $\varphi$-divergence risk measures yield a computationally tractable formulation of the treatment-allocation problem. We begin by introducing the relevant class of risk measures.
	
	\begin{dfn}
		Let $\Phi$ denote the class of proper closed convex functions
		$\varphi:\mathbb{R}\to(-\infty,+\infty]$ satisfying
		$\varphi(1)=0$, $\inf_t\varphi(t)=0$, and
		$1\in\mathrm{int\,dom}\,\varphi$. Given $\varphi\in\Phi$, $I_\varphi$ is the corresponding $\varphi$-divergence of $\mathbb{Q}$ with respect to
		$\mathbb{P}$:
		$$
		I_\varphi(\mathbb{Q},\mathbb{P})=
		\begin{cases}
			\int_\Omega\varphi\!\left(\frac{d\mathbb{Q}}{d\mathbb{P}}
			\right)d\mathbb{P} & \text{if }\mathbb{Q}\ll\mathbb{P}\\
			+\infty & \text{otherwise.}
		\end{cases}
		$$
	\end{dfn}
	The class of $\varphi$-divergences is broad and includes many commonly used statistical distance measures. For example, when $\varphi(t)=t\log t-t+1$, $I_\varphi(\mathbb{Q},\mathbb{P})$ reduces to the Kullback--Leibler divergence, whereas when $\varphi(t)=(\sqrt{t}-1)^2$, it reduces to the Hellinger divergence. In this section, we consider the associated $\varphi$-risk welfare functional:
	\begin{equation}\label{rho_varphi_representation}
		\rho_\varphi(\mathcal{W}(f,Z))
		=\inf_{\mathbb{Q}\in\mathcal{P}(Z)}
		\{\mathbb{E}_\mathbb{Q}[\mathcal{W}(f,Z)]
		+I_\varphi(\mathbb{Q},\mathbb{P})\}\quad \text{for $f\in \mathcal{F},$} 
	\end{equation}
	
	This class evaluates each allocation rule under the least favorable distribution within a $\varphi$-divergence neighborhood of the benchmark distribution. The divergence penalty $I_\varphi$ limits deviations from the benchmark distribution and thereby controls the degree of robustness. The resulting welfare functional is closely related to the optimized certainty equivalent of \cite{Ben-Tal_Teboulle_2007}, adapted here to the welfare-maximization problem.
	\begin{lem}\label{Robust_Social_Welfare}Let Assumptions~\ref{Assump_MTE}, \ref{Assumption2_OCE_Bounded}, and~\ref{Ass_Y0_Bounded} hold.
		Then, for every $\varphi\in\Phi$, the $\varphi$-risk welfare functional $\rho_\varphi$ defined in \eqref{rho_varphi_representation} satisfies Definition~\ref{Coherent_Def}(A1)--(A3) and is therefore a concave risk measure.
		
	\end{lem}%
	
	This class of concave risk measures is not only computationally tractable, but also encompasses many popular risk functionals which we have previously discussed.
	
	\begin{rem}\label{Rem_varphi_examples}
		Equation~\eqref{rho_varphi_representation} generalizes the variational characterization of \cite{fan2025policylearningalphaexpectedwelfare} by extending $\operatorname{AV@R}_\beta$ with exogenous selection to general $\varphi$-divergence penalties under endogenous selection. In particular, if
		\[
		\varphi_{AV@R_\beta}(t)=
		\begin{cases}
			0, & 0\le t\le 1/\beta,\\
			\infty, & \text{otherwise},
		\end{cases}
		\]
		then \eqref{rho_varphi_representation} reduces to the risk measure
		$\operatorname{AV@R}_\beta$ (Example~\ref{Example1}). Similarly, choosing
		$\varphi(t)=t\log t-t+1$
		yields the entropic risk measure (Example~\ref{Example3}).
	\end{rem}
	
	To obtain a tractable characterization of \eqref{rho_varphi_representation}, let
	$\varphi^\ast$ denote the convex conjugate of $\varphi$:
	\[
	\varphi^\ast(s)
	=
	\sup_{t\in\mathbb{R}_+}\{st-\varphi(t)\}
	=
	\sup_{t\in\operatorname{int}\operatorname{dom}\varphi}
	\{st-\varphi(t)\},
	\]
	where the second equality follows from
	\cite[Corollary~12.2.2]{Rockafellar1970}.\footnote{The function
		$\varphi^\ast$ is closed, proper, and convex, with
		$\operatorname{int}\operatorname{dom}\varphi^\ast=(a,b)$, where
		$a=\lim_{t\to-\infty}\varphi^\ast(t)/t$ and
		$b=\lim_{t\to+\infty}\varphi^\ast(t)/t$.}
	The following result shows that \eqref{rho_varphi_representation} admits an equivalent finite-dimensional representation.
	\begin{prop}\label{Cor_program_varphi}
		Let Assumptions~\ref{Assump_MTE}, ~\ref{Assumption2_OCE_Bounded} and ~\ref{Ass_Y0_Bounded} hold.
		Then the welfare functional in \eqref{rho_varphi_representation} admits the representation
		\begin{equation}\label{program_varphi2}
			\rho_\varphi(\mathcal{W}(f,Z))
			=
			\sup_{\eta\in\mathbb{R}}
			\left\{
			\eta
			-
			\mathbb{E}_{\mathbb{P}}
			\!\left[
			\varphi^{*}\!\left(\eta-\mathcal{W}(f,Z)\right)
			\right]
			\right\}.
		\end{equation}
		Consequently, the planner's problem becomes
		\begin{equation}\label{program_varphi3}
			\sup_{f\in\mathcal{F}}
			\rho_\varphi(\mathcal{W}(f,Z))
			=
			\sup_{f\in\mathcal{F}}
			\sup_{\eta\in\mathbb{R}}
			\left\{
			\eta
			-
			\mathbb{E}_{\mathbb{P}}
			\!\left[
			\varphi^{*}\!\left(\eta-\mathcal{W}(f,Z)\right)
			\right]
			\right\}.
		\end{equation}
	\end{prop}%
	
	Proposition~\ref{Cor_program_varphi} has two immediate implications. First,
	the representation in \eqref{program_varphi2} is an optimized certainty
	equivalent in the sense of \cite{Ben-Tal_Teboulle_2007}. Thus, the proposed
	framework encompasses the optimized certainty equivalent as a special case.
	Second, it reduces the optimization problem in
	\eqref{rho_varphi_representation} from one over $(f,\mathbb{Q})$ to one over
	$(f,\eta)$, thereby yielding a computationally tractable formulation. Moreover, Proposition~\ref{Cor_program_varphi} provides a formal link
	between risk-averse treatment allocation and models of uncertainty and
	ambiguity aversion \citep{GILBOA1989141,Marinaccietal2006}. We extend this result by demonstrating that (\ref{rho_varphi_representation}) also admits a representation additively separable in the expected welfare and dispersion.

	\begin{cor}\label{cor:excess_welfare} Let Assumptions~\ref{Assump_MTE}, \ref{Assumption2_OCE_Bounded}, and~\ref{Ass_Y0_Bounded} hold. Let $\phi(t)=t-\varphi^*(t)$. Then the welfare functional in \eqref{rho_varphi_representation} admits the representation:
		\[
		\rho_\varphi(\mathcal W(f,Z))
		=
		\mathbb E[\mathcal W(f,Z)]
		+
		\psi(\mathcal W(f,Z)),
		\]
		where $
		\psi(\mathcal W(f,Z))
		=
		\sup_{\eta\in\mathbb R}
		\mathbb E\!\left[
		\phi\!\left(\eta-\mathcal W(f,Z)\right)
		\right].
		$
	\end{cor}
	\smallskip
	
	This result relates Problem~\eqref{rho_varphi_representation} to the expected welfare maximization problem of \cite{Sasaki_Ura_2024} through an additive risk adjustment. Specifically, the planner maximizes expected welfare together with the adjustment term $\psi$, which reflects the effect of risk aversion relative to the benchmark level $\eta^\star$. When the social planner is risk neutral, then criterion reduces to expected welfare maximization. As shown below, this property plays a key role in the regret analysis.
	
	\subsection{Distributional Welfare and the Kusuoka Representation}
	
	We now show that a broad class of coherent risk measures admits a representation as a mixture of $\operatorname{AV@R}_\beta$ functionals. This requires one additional assumption on the risk functional beyond those used so far: $\rho$ is law invariant, so that a Kusuoka representation can be derived. We emphasize that this assumption is not needed for the results of the previous section---it is invoked here solely to obtain the mixture representation.

	\textcolor{black}{Recall that $(\Omega,\mathcal{B},\mathbb{P})$ is the probability space supporting
		$(X, Z, U_D, Y_0, Y_1)$. Since $U_D \sim U(0,1)$, the space
		$(\Omega,\mathcal{B},\mathbb{P})$ is atomless.\footnote{Atomlessness is a property of the
			underlying space, not of the random variables defined on it. In particular,
			$\mathcal{W}(f,Z)$ integrates out $U_D$ and is therefore
			$\sigma(X,Z)$-measurable; when $(X,Z)$ is discrete, $\mathcal{W}(f,Z)$ has a
			discrete distribution. The representation still applies because $\rho$ is
			defined on $\mathcal{Z}=L_p(\Omega,\mathcal{B},\mathbb{P})$. It could fail if $\rho$ were
			defined only on $L_p(\Omega,\sigma(X,Z),\mathbb{P})$, which has atoms when $(X,Z)$
			is discrete. Because $\rho$ is law-invariant, working on an atomless space is a
			normalization rather than a substantive restriction.} Under
		Assumptions~\ref{Assump_MTE}, \ref{Assumption2_OCE_Bounded},
		and~\ref{Ass_Y0_Bounded}, $\mathcal{W}(f,Z)\in L^\infty(\Omega,\mathcal{B},\mathbb{P})\subset\mathcal{Z}$,
		so the Kusuoka representation applies \citep{Shapiro_Kusuoka2013}.}
	
	The following theorem establishes the corresponding characterization of the planner's welfare criterion.
	\begin{thm}\label{Kusuoka_Welfare}
		Let Assumptions~\ref{Assump_MTE}, \ref{Assumption2_OCE_Bounded}, and~\ref{Ass_Y0_Bounded} hold. Suppose that
		$\rho:\mathcal{Z}\to\mathbb{R}$ is a proper, upper semicontinuous,
		law-invariant risk functional. Let $\mathcal{M}$ denote the set of
		probability measures on $(0,1]$. Then the following statements hold.
		
		\begin{itemize}
			\item[(i)] If $\rho$ satisfies (A1),(A2), and (A3), then
			\begin{equation}\label{Convex_Rep_Program}
				\sup_{f\in\mathcal{F}}
				\rho(\mathcal{W}(f,Z))
				=
				\sup_{f\in\mathcal{F}}
				\inf_{\mu\in\mathcal{M}}
				\left\{
				\int_0^1
				\operatorname{AV@R}_\alpha(\mathcal{W}(f,Z))
				\,\mu(d\alpha)
				+
				c(\mu)
				\right\},
			\end{equation}
			where
			\[
			c(\mu)
			\triangleq
			\sup_{\tilde Z\in\mathcal{A}_\rho}
			\left\{
			-
			\int_0^1
			\operatorname{AV@R}_\alpha(\tilde Z)
			\,\mu(d\alpha)
			\right\},
			\]
			and
			\[
			\mathcal{A}_\rho
			\triangleq
			\{\tilde Z\in\mathcal{Z}:\rho(\tilde Z)\ge0\}
			\]
			is the acceptance set of $\rho$.
			
			\item[(ii)] If, in addition, $\rho$ satisfies (A4), then
			\begin{equation}\label{Coh_Convex_Rep_Program}
				\sup_{f\in\mathcal{F}}
				\rho(\mathcal{W}(f,Z))
				=
				\sup_{f\in\mathcal{F}}
				\inf_{\mu\in\mathcal{M}_\rho}
				\int_0^1
				\operatorname{AV@R}_\alpha(\mathcal{W}(f,Z))
				\,\mu(d\alpha),
			\end{equation}
			where $\mathcal{M}_\rho\triangleq\{\mu\in\mathcal{M}:c(\mu)=0\}$. Since $\mathcal{A}_\rho$ is a cone under (A4), $c(\mu)\in\{0,+\infty\}$ for every $\mu\in\mathcal{M}$.
		\end{itemize}
	\end{thm}%
	
	The resulting characterization shows that $\operatorname{AV@R}_\beta$ is the fundamental building block underlying a wide class of risk-averse welfare criteria: the planner's welfare criterion can be expressed as a mixture of $\operatorname{AV@R}_\alpha$ functionals, where the mixing measure summarizes the planner's sensitivity to different regions of the welfare distribution. In the coherent case, the representation simplifies to the classical Kusuoka representation as an infimum purely over mixtures of $\operatorname{AV@R}_\alpha$ functionals.
	
	Economically, the theorem provides a distributional interpretation of risk-averse welfare maximization. Rather than requiring the planner to specify a priori which quantiles or disadvantaged groups should receive priority, the planner's preferences are summarized by the weights assigned to different regions of the welfare distribution. Consequently, optimization under any law-invariant coherent risk measure reduces to optimization over mixtures of $\operatorname{AV@R}_\alpha$ functionals, yielding a tractable characterization of the planner's objective.
	
	Combined with Proposition~\ref{CVAR_cor}, this representation forms the basis for the finite-sample regret analysis in Section~\ref{s4}. Although the Kusuoka representation is well established in mathematical finance and stochastic programming \citep{Shapiro_Kusuoka2013,NoyanRudolf2015,Pichleretal2012}, to the best of our knowledge this is its first application to the treatment allocation problem. This result provides the foundation for a connection to second-order stochastic dominance (Proposition~\ref{Prop_SOSD} below) and for the empirical analysis in Section~\ref{sec:empirical}, where risk-averse policies improve welfare among disadvantaged subpopulations without requiring the planner to identify ex ante which groups or quantiles should receive priority.
	
	\begin{prop}\label{Prop_SOSD}
		Let Assumptions~\ref{Assump_MTE}, \ref{Assumption2_OCE_Bounded}, and~\ref{Ass_Y0_Bounded} hold. Consider $f_1,f_2 \in \mathcal{F}.$ The following statements are equivalent
		\begin{enumerate}
			\item The distribution of welfare induced by $f_1$ second-order stochastically dominates that induced by $f_2$: $\mathcal{W}(f_1,Z)
			\succeq_2
			\mathcal{W}(f_2,Z)$
			\item $\rho(\mathcal{W}(f_1,Z)) \geq \rho(\mathcal{W}(f_2,Z)) $ for every law-invariant, proper, upper semicontinuous $\rho$ which satisfies (A1),(A2), and (A3).
		\end{enumerate}
	\end{prop}%
	
	Proposition \ref{Prop_SOSD} demonstrates that the choice of risk functional is irrelevant whenever one policy's welfare distribution second-order stochastically dominates another's: so the ranking is robust to the specification of $\rho$. However, when distributions cross, the choice of $\rho$ is what resolves the comparison — it encodes the planner's trade-off between lower-tail and average welfare, playing a similar role as the inequality-aversion parameter in \cite{atkinson1970measurement}.
	
	\section{Regret Guarantees}\label{s4}
	
	Having characterized the population problem, we now turn to finite-sample learning.\footnote{Throughout this section we assume the class $\mathcal{F}$ is pointwise measurable \citep[Ex.~2.3.4]{VaartWellner1996}, so that the suprema below and the empirical maximizer $\hat f_n$ are measurable. This is a regularity condition satisfied by the policy classes used in practice and imposes no substantive restriction; see \citet{Kiatagawa_Tetenov_2018}.} The population objective $\rho\bigl(\mathcal{W}(f,Z)\bigr)$ depends on the unknown distribution of $(Y_0,Y_1,D,Z)$ and is infeasible to implement directly, so the planner instead optimizes the empirical objective computed from an i.i.d.\ sample; our goal is to bound the resulting welfare loss and establish finite-sample guarantees for the empirical policy.
	
	Given an i.i.d.\ sample $\{Z_i\}_{i=1}^n$, let $\mathbb{P}_n =n^{-1}\sum_{i=1}^n\delta_{Z_i}$
	denote the empirical distribution, and define the empirical conditional average outcome by
	\[
	\mathcal{W}_n(f,Z)
	\triangleq
	\mathbb{E}_{\mathbb{P}_n}[Y(\pi)\mid Z].
	\]
	\textcolor{black}{Throughout this section we treat the conditional welfare function $\mathcal{W}(f,\cdot)$ as known, so that the source of sampling variation is the empirical distribution of $Z$. Extending the analysis to estimated MTE components requires accounting for first-stage estimation error and is left to future work.}
	
	The \emph{empirical risk-averse maximizer} is
	\begin{equation}\label{ERAM_General}
		\hat{f}_n
		\triangleq
		\arg\max_{f\in\mathcal{F}}
		\rho\bigl(\mathcal{W}_n(f,Z)\bigr).
	\end{equation}
	
	\textcolor{black}{We assume the empirical maximum is attained, which holds automatically when, for example, $\Pi$ is a VC class.} The corresponding \emph{risk-averse regret} is
	\begin{equation}\label{Regret_General}
		\mathcal{R}_\rho(\hat{f}_n)
		\triangleq
		\sup_{f\in\mathcal{F}}
		\rho\bigl(\mathcal{W}(f,Z)\bigr)
		-
		\rho\bigl(\mathcal{W}(\hat{f}_n,Z)\bigr).
	\end{equation}
	When $\rho=\mathbb{E}$, \eqref{ERAM_General} and \eqref{Regret_General} reduce to the empirical welfare maximization problem and regret criterion studied by \cite{Manski2004} and \cite{Kiatagawa_Tetenov_2018}.
	
	Bounding this regret requires two ingredients: a measure of the statistical complexity of the induced welfare class $\mathcal{F}$, and a regularity condition on the risk functional $\rho$ itself. We introduce each in turn.
	
	\subsection{Rademacher Complexity}\label{sec:rademacher_complexity}

	The empirical Rademacher complexity measures the extent to which a function class can correlate with random sign noise, and provides the standard data-dependent measure of the statistical complexity of $\mathcal{F}$ \citep{Shalev-Shwartz_Ben-David_2014}. It is defined as
	
	\begin{equation}\label{Rademacher_def}
		\mathfrak{R}_n(\mathcal{F};Z^n)
		\triangleq
		\mathbb{E}_{\sigma^n}
		\left[
		\sup_{f\in\mathcal{F}}
		\frac{1}{n}
		\sum_{i=1}^n
		\sigma_i f(Z_i)
		\right],
	\end{equation}
	where $\{\sigma_i\}_{i=1}^n$ are i.i.d.\ Rademacher random variables satisfying $
	\mathbb{P}(\sigma_i=1)
	=\\
	\mathbb{P}(\sigma_i=-1)
	=
	\frac{1}{2},
	$
	and $\sigma^n=(\sigma_1,\ldots,\sigma_n)$.
	
	\subsection{A Regularity Condition on the Risk Functional}\label{sec:lipschitz_condition}
	
	The second ingredient is a regularity condition on $\rho$. Every coherent risk measure is Lipschitz with constant $L=1$ under the sup norm, which guarantees continuity of $\rho$ on $L^\infty$ but is too coarse to drive a finite-sample rate: the empirical criterion $\rho(\mathcal{W}_n(f,Z))$ evaluates $\rho$ under the empirical measure $\mathbb{P}_n$ rather than perturbing a fixed random variable under $\mathbb{P}$, so what governs the rate instead is the variational representation of $\rho$ developed in Section~\ref{s3}. For $\varphi$-risk functionals this yields the risk-measure-specific constant $L_\varphi$ (Table~\ref{tab:lipschitz}); for general coherent measures reached through the Kusuoka representation, it is the reciprocal of the minimal $\alpha$ in its support. Corollary~\ref{Cor_specific} reports the $L_\varphi$ for the risk measures considered here.
	
	\smallskip
	
	\begin{rem}\label{rem:diff_convention}
		Throughout this section we assume that $\varphi^*$ is differentiable on $\R$. This is for expositional convenience only: since $\varphi^*$ is convex, its one-sided derivatives exist everywhere and are non-decreasing, and all results below hold verbatim with $\varphi^{*\prime}$ replaced by the right derivative $\varphi^{*\prime}_+$, under the weaker requirement that $\varphi^*$ be Lipschitz on the relevant range. In particular, the results apply to $\varphi_{AV@R_\beta}$, whose conjugate is Lipschitz with constant $1/\beta$ though not everywhere differentiable.
	\end{rem}

	\begin{table}[h]
		\centering
		\begin{tabular}{lll}
			\toprule
			Risk measure $\rho$ & Definition & Regret-rate constant $L_\varphi$ \\
			\midrule
			AV$@$R$_\beta$ & $\frac{1}{\beta}\int_0^\beta F_T^{-1}(\tau)\,d\tau$
			& $1/\beta$, $\beta\in(0,1]$ \\[4pt]
			Mean--semideviation, $p=1$\footnotemark & $\mathbb{E}[T]-c\,\mathbb{D}_1^-[T]$
			& \tiny{$2/\alpha_{\min}$, $\alpha_{\min}=\inf_f\Pr\{\mathcal{W}(f,Z)<\EE[\mathcal{W}(f,Z)]\}$} \\[4pt]
			Entropic risk  & $-\frac{1}{\gamma}\ln\mathbb{E}[e^{-\gamma T}]$
			& $e^{2\gamma(B+M)}$ \\[4pt]
			Risk-neutral ($\rho=\mathbb{E}$) & $\mathbb{E}[T]$
			& $1$ \\
			\bottomrule
		\end{tabular}
		\caption{Regret-rate constants for leading risk measures: the multiplicative factor $L_\varphi$ by which each measure's finite-sample regret bound scales relative to the risk-neutral rate, as demonstrated in Corollary~\ref{Cor_specific} and Appendix C. This is \emph{not} the sup-norm Lipschitz modulus which equals $1$ for every coherent risk measure and therefore does not distinguish between distinct risk measures.}
		\label{tab:lipschitz}
	\end{table}
	\footnotetext{Stated for $p=1$ and $c >0$; the $p\neq1$ case is genuinely open. The $p=1$ and $c >0$ case also requires an additional local regularity assumption for verification. (see Appendix~C)}
	
	\subsection{Contraction and Symmetrization}\label{sec:contraction_symmetrization}
	
	Because $\rho_\varphi$ is a nonlinear functional of the welfare distribution, it cannot be handled by the pointwise contraction arguments used in risk-neutral policy learning. The variational representation of Proposition~\ref{Cor_program_varphi} restores tractability by replacing the risk functional with an expectation over the convex conjugate $\varphi^*$, to which the classical contraction principle applies. We first show that the optimization over the dual variable $\eta$ can be confined to a bounded interval, and then bound the Rademacher complexity of the transformed class $\varphi^*\circ\mathcal{G}$ by that of the induced welfare class $\mathcal{F}$.
	
	\begin{lem}[Range Restriction]\label{Lem_OCE_Range}
		Let $\varphi\in\Phi$. Suppose that $Z\in\mathcal{Z}$ satisfies $a\le Z\le b$ almost surely, for constants $a\le b$. Then
		\begin{equation}\label{Range_ineq}
			\sup_{\eta\in\mathbb{R}}
			\left\{
			\eta
			-
			\mathbb{E}\!\left[\varphi^*(\eta-Z)\right]
			\right\}
			=
			\sup_{\eta\in[a,b]}
			\left\{
			\eta
			-
			\mathbb{E}\!\left[\varphi^*(\eta-Z)\right]
			\right\}.
		\end{equation}
	\end{lem}%
	
	Our key Assumptions ~\ref{Assump_MTE},~\ref{Assumption2_OCE_Bounded} and \ref{Ass_Y0_Bounded} ensure that, for every $f\in\mathcal{F}$, $\mathcal{W}(f,Z)\in[a,b]$ almost surely,
	where $a=-(B+M),$ $b=B+M,$ and we define $K\triangleq b-a=2(B+M).$\footnote{The constant $K$ is specific to the finite-sample analysis developed in this section and should not be confused with the identification constant $K_0$ introduced in Remark~\ref{Rem_K}.} The same bounds apply under the empirical distribution: since $\mathbb{P}_n$ is supported on the sample points $\{Z_i\}_{i=1}^n$ and $\mathcal{W}(f,Z_i)\in[a,b]$ for each $i$, we also have $\mathcal{W}_n(f,Z)\in[a,b]$ almost surely for every $f\in\mathcal{F}$ and every $n$. Lemma~\ref{Lem_OCE_Range} may therefore be applied to the population and empirical objects with the same restriction $\eta\in[a,b]$.
	\begin{lem}[OCE Contraction]\label{Lem_Contraction}
		Let Assumptions~\ref{Assump_MTE}, \ref{Assumption2_OCE_Bounded}, and~\ref{Ass_Y0_Bounded} hold, and $\varphi\in\Phi$. Define the Lipschitz constant over the restricted space: 
		\[
		L_\varphi
		\triangleq
		\sup_{s\in[-K,K]}
		\varphi^{*\prime}(s)
		<
		\infty.
		\]
		Then
		\begin{equation}\label{Contraction_ineq}
			\mathfrak{R}_n(\{\varphi^*\circ g:g\in\mathcal G\};Z^n)
			\le
			L_\varphi\,
			\mathfrak{R}_n(\mathcal G;Z^n)
			\le
			L_\varphi
			\left(
			\mathfrak{R}_n(\mathcal F;Z^n)
			+
			\frac{K}{2\sqrt n}
			\right).
		\end{equation}
		\noindent where 
		$\mathcal{G}
		\triangleq
		\{\eta-\mathcal{W}(f,\cdot):
		f\in\mathcal{F},\;
		\eta\in[a,b]\}$.
	\end{lem}%
	
	The contraction argument applies to the scalar convex conjugate $\varphi^*$ --- Lipschitz on $[-K,K]$ --- rather than to $\rho_\varphi$ itself, which is a nonlinear functional of the distribution of $\mathcal{W}(f,Z)$ and cannot be evaluated pointwise.
	\begin{lem}[Symmetrization for $\varphi$-Divergence Risk Measures]\label{Lem_Symmetrization}
		Let Assumptions~\ref{Assump_MTE},~\ref{Assumption2_OCE_Bounded}, and \ref{Ass_Y0_Bounded} hold. Let
		$\rho=\rho_\varphi$, where $\varphi\in\Phi$ has the associated restricted Lipschitz constant $L_\varphi$ as given by Lemma~\ref{Lem_Contraction}. Then
		\begin{equation}\label{Symmetrization_ineq}
			\mathbb{E}
			\left[
			\sup_{f\in\mathcal F}
			\left\{
			\rho_\varphi(\mathcal W(f,Z))
			-
			\rho_\varphi(\mathcal W_n(f,Z))
			\right\}
			\right]
			\le
			2L_\varphi
			\left(
			\mathbb{E}\!\left[
			\mathfrak R_n(\mathcal F;Z^n)
			\right]
			+
			\frac{K}{2\sqrt n}
			\right).
		\end{equation}
	\end{lem}%
	
	The proof symmetrizes the empirical average $\mathbb{E}_{\mathbb{P}_n}[\varphi^*\circ g]$ rather than the nonlinear functional $\rho_\varphi(\mathcal{W}_n(f,Z))$ directly, since the latter would require an unbiasedness property that generally fails for nonlinear risk functionals, whereas the former is a genuine empirical average and admits classical symmetrization --- the standard approach for optimized certainty equivalents \citep[Lemma~3]{lee2020learning}. Since $\operatorname{AV@R}_\beta$ and the entropic risk measure are both $\varphi$-divergence risk measures (Remark~\ref{Rem_varphi_examples}), the lemma applies to both immediately, with constants reported in Corollary~\ref{Cor_specific}.
	\subsection{Extension to General Coherent Risk Measures via the Kusuoka Representation}\label{s4_kusuoka}
	
	For general law-invariant coherent risk measures, Theorem~\ref{Kusuoka_Welfare}(ii) reduces the problem to a family of $\operatorname{AV@R}_\alpha$ functionals via the Kusuoka representation; since each $\operatorname{AV@R}_\alpha$ is itself a $\varphi$-divergence risk measure with a globally Lipschitz conjugate (Remark~\ref{Rem_varphi_examples}), Lemma~\ref{Lem_Symmetrization} applies to every component without further assumptions. The resulting bound is therefore governed by the Kusuoka mixing measure's weight on small $\alpha$: greater weight on the lower tail yields a more conservative criterion, but also larger estimation error.
	\begin{ass}[Kusuoka Tail Control]\label{Ass_KusuokaTail}
		There exists a constant $\alpha_{\min}\in(0,1]$ such that, for every
		$f\in\mathcal{F}$ and every sample size $n$, every minimizing (or
		asymptotically minimizing) probability measure $\mu$ in the Kusuoka
		representation~\eqref{Coh_Convex_Rep_Program}, evaluated at either
		$\mathcal{W}(f,Z)$ or $\mathcal{W}_n(f,Z)$, satisfies
		\[
		\operatorname{supp}(\mu)\subseteq[\alpha_{\min},1].
		\]
	\end{ass}
	\textcolor{black}{Assumption~\ref{Ass_KusuokaTail} is data dependent; a sufficient alternative is that $\rho$'s Kusuoka representation~\eqref{Coh_Convex_Rep_Program} itself use only mixing measures supported on $[\alpha_{\min},1]$, which holds for $\operatorname{AV@R}_\beta$ ($\alpha_{\min}=\beta$) and, more generally, spectral risk measures with bounded risk spectrum \citep{Acerbi2002}; see Online Appendix C for the mean-semideviation case.}
	\begin{prop}[Kusuoka Symmetrization]\label{Prop_Kusuoka_Symmetrization}
		Let Assumptions~\ref{Assump_MTE},
		\ref{Assumption2_OCE_Bounded},
		\ref{Ass_Y0_Bounded}, and~\ref{Ass_KusuokaTail} hold. Let
		$\rho:\mathcal{Z}\to\mathbb{R}$ be a coherent, law-invariant risk measure.
		Then
		\begin{equation}\label{Kusuoka_Symmetrization_ineq}
			\mathbb{E}
			\left[
			\sup_{f\in\mathcal F}
			\left\{
			\rho(\mathcal W(f,Z))
			-
			\rho(\mathcal W_n(f,Z))
			\right\}
			\right]
			\le
			\frac{2}{\alpha_{\min}}
			\left(
			\mathbb{E}\!\left[
			\mathfrak R_n(\mathcal F;Z^n)
			\right]
			+
			\frac{K}{2\sqrt n}
			\right)
			+
			\epsilon_n(\alpha_{\min}),
		\end{equation}
		where
		\[
		\epsilon_n(\alpha) = \frac{K}{\alpha}\sqrt{\frac{2\pi}{n}} = O(n^{-1/2})
		\]
	\end{prop}%

	\begin{rem}
		Proposition~\ref{Prop_Kusuoka_Symmetrization} shows that, for general law-invariant coherent risk measures, the finite-sample bound depends on $1/\alpha_{\min}$---the smallest tail level receiving positive weight in the Kusuoka representation. This admits a natural interpretation: placing greater weight on increasingly adverse tail outcomes makes the welfare criterion more conservative, but also makes uniform estimation from finite samples more difficult. Consequently, stronger tail aversion entails a larger statistical complexity.
		
		When $\rho=\operatorname{AV@R}_\beta$, we have $\alpha_{\min}=\beta$, and the bound in \eqref{Kusuoka_Symmetrization_ineq} reduces to that of Lemma~\ref{Lem_Symmetrization}, since $L_\varphi=1/\beta$. Thus, the Kusuoka and $\varphi$-divergence analyses coincide whenever both representations are available.
	\end{rem}
	\subsection{Uniform Convergence}\label{sec:UC}
	
	Combining Lemmas~\ref{Lem_Contraction} and~\ref{Lem_Symmetrization} with Proposition~\ref{Prop_Kusuoka_Symmetrization}, we establish uniform convergence of the empirical welfare criterion to its population counterpart. This uniform convergence immediately yields finite-sample guarantees for the empirical risk-averse policy.
	
	\begin{prop}[Uniform Convergence]\label{Prop_UC}
		Let Assumptions~\ref{Assump_MTE},
		\ref{Assumption2_OCE_Bounded}, and~\ref{Ass_Y0_Bounded}
		hold.
		
		\textnormal{(i) $\varphi$-divergence risk measures:}
		Let
		$\rho=\rho_\varphi$, where $\varphi\in\Phi$ satisfies the conditions of
		Lemma~\ref{Lem_Symmetrization} with Lipschitz constant $L_\varphi$.
		Then, for every $\delta\in(0,1]$,
		\begin{equation}\label{UC_bound}
			\sup_{f\in\mathcal F}
			\left|
			\rho_\varphi(\mathcal W(f,Z))
			-
			\rho_\varphi(\mathcal W_n(f,Z))
			\right|
			\le
			L_\varphi
			\left(
			4\,\mathbb E[\mathfrak R_n(\mathcal F;Z^n)]
			+\frac{2K}{\sqrt n}
			+
			K\sqrt{\frac{\log(1/\delta)}{2n}}
			\right),
		\end{equation}
		with probability at least $1-\delta$.
		
		\medskip
		
		\textnormal{(ii) General coherent, law-invariant risk measures:}
		Suppose, in addition, that Assumption~\ref{Ass_KusuokaTail} holds, and let
		$\rho:\mathcal Z\to\mathbb R$ be a coherent, law-invariant risk functional.
		Then, for every $\delta\in(0,1]$,
		\begin{equation}\label{UC_bound_Kusuoka}
			\sup_{f\in\mathcal F}
			\left|
			\rho(\mathcal W(f,Z))
			-
			\rho(\mathcal W_n(f,Z))
			\right|
			\le
			\frac{1}{\alpha_{\min}}
			\left(
			4\,\mathbb E[\mathfrak R_n(\mathcal F;Z^n)]
			+\frac{2K}{\sqrt n}
			+
			K\sqrt{\frac{\log(1/\delta)}{2n}}
			\right)
			+
			2\epsilon_n,
		\end{equation}
		with probability at least $1-\delta$.
	\end{prop}
	\proof See Appendix~A. \eproof
	
	\begin{rem}
		The bounds depend on the Rademacher complexity $\mathfrak{R}_n(\mathcal{F};Z^n)$ plus an $O(n^{-1/2})$ term from the variational representation's optimization variables ($\eta$, and $\alpha$ under Kusuoka); risk aversion enters only through $L_\varphi$ or $1/\alpha_{\min}$ (Corollary~\ref{Cor_specific}). The proposition applies to any policy class admitting a suitable Rademacher bound, with finite-VC classes as the leading example.
	\end{rem}
	\subsection{Regret Bounds}\label{sec:regret}
	We now establish the main finite-sample regret bound for the empirical risk-averse policy.
	
	\begin{thm}[Finite-Sample Regret Bounds]\label{Thm_Regret}
		
		\textnormal{(i) $\varphi$-divergence risk measures.}
		Under the hypotheses of Proposition~\ref{Prop_UC}(i), the empirical
		risk-averse maximizer $\hat f_n$ satisfies, with probability at least
		$1-\delta$,
		\begin{equation}\label{Regret_bound}
			\mathcal{R}_{\rho_\varphi}(\hat f_n)
			\le
			2L_\varphi
			\left(
			4\,\mathbb E[\mathfrak R_n(\mathcal F;Z^n)]
			+\frac{2K}{\sqrt n}
			+
			K\sqrt{\frac{\log(1/\delta)}{2n}}
			\right).
		\end{equation}
		
		\medskip
		
		\textnormal{(ii) General coherent, law-invariant risk measures.}
		Under the hypotheses of Proposition~\ref{Prop_UC}(ii), the empirical
		risk-averse maximizer satisfies, with probability at least $1-\delta$,
		\begin{equation}\label{Regret_bound_AVAR}
			\mathcal{R}_{\rho}(\hat f_n)
			\le
			\frac{2}{\alpha_{\min}}
			\left(
			4\,\mathbb E[\mathfrak R_n(\mathcal F;Z^n)]
			+\frac{2K}{\sqrt n}
			+
			K\sqrt{\frac{\log(1/\delta)}{2n}}
			\right)
			+
			4\epsilon_n(\alpha_{\min}).
		\end{equation}
	\end{thm}
	
	\proof See Appendix~A. \eproof
	
	\begin{rem}[Rademacher vs.\ Vapnik-Chervonenkis (VC) complexity]\label{Rem_Rademacher_VC}
		Theorem~\ref{Thm_Regret} is stated in terms of Rademacher
		complexity, which is data-dependent and applies to any policy
		class, including ones with infinite VC dimension. When $\Pi$ is a
		VC class of dimension $V$ --- the standard requirement in the
		policy-learning literature \citep{Kiatagawa_Tetenov_2018,AtheyWager2021}, covering threshold rules, linear
		classifiers, regression trees, and intersections of halfspaces ---
		the bound
		$\mathfrak{R}_n(\mathcal{F},Z^n)\leq
		M\sqrt{2V\log(en/V)/n}$
		\citep[Theorem~26.5]{Shalev-Shwartz_Ben-David_2014} recovers
		VC rates as a special case, nesting existing policy-learning
		results within Theorem~\ref{Thm_Regret}.
	\end{rem}
	Proposition~\ref{Prop_UC} and Theorem~\ref{Thm_Regret} are stated in terms of the induced welfare class $\mathcal{F}$ rather than the underlying policy class $\Pi$. This formulation is natural in the MTE framework, where each policy rule $\pi\in\Pi$ induces a welfare function through an MTE-weighted integral. Consequently, the complexity of $\mathcal{F}$ is inherited from the complexity of the underlying policy class $\Pi$ together with the policy-to-welfare mapping.
	
	The following corollary specializes Theorem~\ref{Thm_Regret} to the risk measures considered in this paper. For notational convenience, define
	\[
	\mathcal{R}_n
	\triangleq
	\mathbb{E}\!\left[\mathfrak{R}_n(\mathcal{F};Z^n)\right]
	\]
	and
	\[
	Q_n(\delta)
	\triangleq
	4\mathcal{R}_n
	+\frac{2K}{\sqrt n}
	+
	K\sqrt{\frac{\log(1/\delta)}{2n}},
	\]
	which is the common complexity term appearing in Theorem~\ref{Thm_Regret}.
	
	\begin{cor}[Risk-Measure-Specific Regret Bounds]\label{Cor_specific}
		Under the assumptions of Theorem~\ref{Thm_Regret}, the following bounds hold with probability at least $1-\delta$.
		
		\begin{enumerate}
			\item[(i)] $\boldsymbol{AV@R_\beta}$.
			Since $L_\varphi=1/\beta$,
			\[
			\mathcal{R}_{\operatorname{AV@R}_\beta}(\hat f_n)
			\le
			\frac{2}{\beta}\,
			Q_n(\delta).
			\]
			
			\item[(ii)] \textbf{Entropic risk.}
			For
			\[
			\varphi(t)
			=
			\frac{1}{\gamma}
			(t\log t-t+1),
			\]
			the corresponding Lipschitz constant is
			\[
			L_\varphi
			=
			e^{2\gamma(B+M)},
			\]
			and therefore
			\[
			\mathcal{R}_{\rho}(\hat f_n)
			\le
			2e^{2\gamma(B+M)}
			Q_n(\delta).
			\]
			
			\item[(iii)] \textbf{Risk-neutral welfare.}
			Since $\rho=\mathbb{E}$ corresponds to the case $\beta=1$,
			\[
			\mathcal{R}_{\rho}(\hat f_n)
			\le
			2Q_n(\delta).
			\]
		\end{enumerate}
	\end{cor}
	
	Corollary~\ref{Cor_specific} highlights two patterns. The $\operatorname{AV@R}_\beta$ bound deteriorates as $\beta\downarrow0$ --- policies weighting extreme lower-tail outcomes are harder to learn --- with an exact constant $1/\beta$, since $\varphi_{\operatorname{AV@R}_\beta}^*$ is globally Lipschitz and needs no domain truncation. The entropic-risk constant $L_\varphi=e^{2\gamma(B+M)}$ depends exponentially on both the risk-aversion parameter $\gamma$ and the welfare range $B+M$, so learning becomes markedly harder as the planner weights adverse outcomes more heavily.
	
	At $\beta=1$, $\operatorname{AV@R}_1$ is expected welfare and part~(iii) recovers the regret guarantee of \cite{Kiatagawa_Tetenov_2018}, extended here to endogenous treatment selection via the MTE framework; the mean--semideviation case (Appendix~C) instead uses the Kusuoka representation, since $\mathbb{D}_p^{-}$ has no direct $\varphi$-divergence representation.
	\begin{rem}[Comparison with \citet{Kiatagawa_Tetenov_2018}]\label{Rem_KT_comparison}
		Part~(iii) recovers \citet{Kiatagawa_Tetenov_2018}'s risk-neutral, exogenous-selection bound of order $M\sqrt{V/n}$ (via Dudley's entropy integral). Applying the VC-Rademacher bound from Remark~\ref{Rem_Rademacher_VC} to Theorem~\ref{Thm_Regret} yields
		\[
		M\sqrt{\frac{V\log(en/V)}{n}},
		\]
		matching their minimax-optimal rate up to a logarithmic factor, while extending it to endogenous treatment selection and general risk-averse welfare criteria.
	\end{rem}
	\section{Empirical Application: Risk-Averse Treatment Allocation for College Enrollment}\label{sec:empirical}
	
	We illustrate the proposed framework using the college-proximity data of \citet{Card1995}, in which college attendance is endogenous and a valid instrumental variable is available. Holding the identification strategy and policy class fixed, we examine how the optimal treatment allocation changes as the planner's objective shifts from expected welfare to a $AV@R$ maximization. This comparison isolates the effect of risk aversion on policy design. As predicted by the theory, risk-averse objectives induce substantial reallocations relative to the risk-neutral benchmark by assigning greater weight to lower-welfare subpopulations.
	\subsection{Data and Policy Problem}\label{sec:emp_setup}
	The empirical analysis uses the NLS Young Men cohort studied by \citet{Card1995}, consisting of $n=3{,}010$ observations after excluding individuals with missing log wages. The outcome is log hourly wages in 1976, the treatment is college completion ($D=\mathbf{1}\{\text{years of schooling}\ge16\}$), and the instrument is an indicator for the presence of a four-year college in the individual's county of residence in 1966 (\texttt{nearc4}). Although this instrument is weak ($F \approx 3.9$), it is well studied, allowing us to isolate how our framework departs systematically from prior results in a familiar setting.  The covariates include potential experience and its square, race, indicators for SMSA and southern residence (1966 and 1976), and eight regional indicators. In the estimation sample, $27.1\%$ of individuals complete college, $23.4\%$ are Black, and $40.4\%$ reside in the South.
	
	The planner allocates college subsidies subject to a budget constraint that permits treatment of at most one-half of the population:
	\[
	\Pi_{0.5}
	=
	\left\{
	\pi:\mathcal{Z}\rightarrow\{0,1\}
	:
	\mathbb{E}_n[\pi(Z)]\le0.5
	\right\}.
	\]
	
	Candidate policies are evaluated either by expected welfare,
	$\mathbb{E}[\mathcal{W}(f,Z)]$, or by the risk-averse criterion
	$\operatorname{AV@R}_{\beta}(\mathcal{W}(f,Z))$. Because the budget constraint requires the planner to rank individuals rather than simply identify those with positive treatment effects, the choice of welfare criterion directly affects the ranking and, consequently, the optimal allocation.
	
	\subsection{MTE Estimation}\label{sec:emp_mte}
	
	We estimate the marginal treatment effect (MTE) using the local instrumental variables approach of \citet{Heckman_Vytlacil_2005}. The propensity score is estimated by a probit of college completion on the instrument and covariates, and log wages are regressed on the covariates and a quartic polynomial in the estimated propensity score, allowing interactions with race and potential experience.\footnote{Propensity scores are clipped to $[0.05, 0.95]$ before entering
		the polynomial-in-propensity outcome regression, guarding against the instability such
		specifications are known to exhibit near the boundary of the unit interval.} The MTE is obtained by differentiating the estimated outcome equation with respect to the propensity score and calibrating it to the 2SLS estimate ($\hat{\beta}_{\mathrm{2SLS}}=1.39$). Individual treatment gains are then computed as $\widehat{\Delta}_i=\widehat{\mathrm{MTE}}(\hat v(Z_i))$.
	
	Figure~\ref{Fig_MTE} reports the estimated MTE. Two features are noteworthy. First, the MTE declines with $\varepsilon$, consistent with selection on gains. Second, the MTE for Black individuals lies uniformly above that for non-Black individuals (average gains of $1.08$ versus $0.62$ log points), while their untreated welfare is substantially lower ($6.02$ versus $6.48$ log points). Thus, individuals with the lowest baseline welfare also have the largest expected treatment gains, creating significant scope for risk-averse targeting. Since only $4\%$ of individuals have negative estimated gains, the $50\%$ budget constraint requires a planner to determine which positively affected individuals receive treatment, not target them all.
	\begin{figure}[t]
		\centering
		\includegraphics[width=\textwidth]{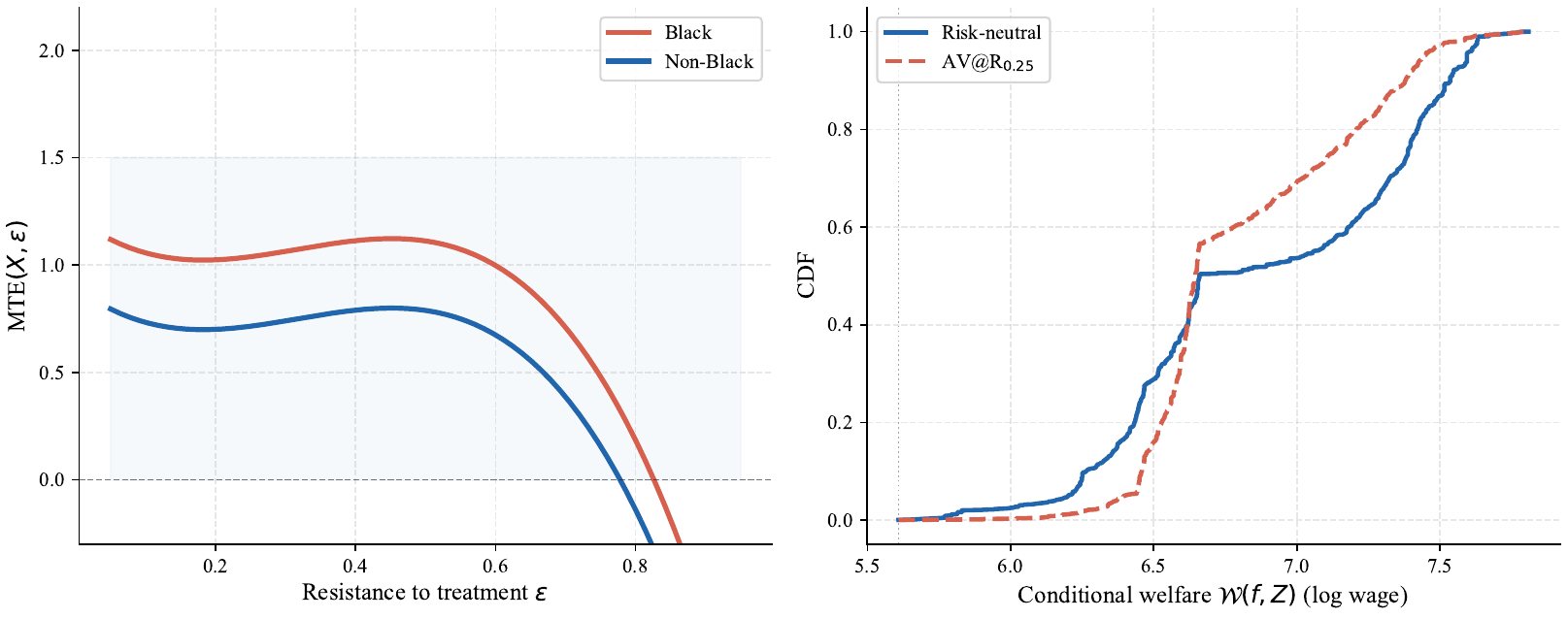}
		\caption{Selection on gains and distributional welfare accounting. Left: $\widehat{\text{MTE}}(X,\varepsilon)$ by race; the downward slope in $\varepsilon$ indicates that individuals most likely to select into treatment gain the most, confounding naive treatment-effect comparisons. Right: CDF of conditional welfare under the risk-neutral and $\text{AV@R}_{0.25}$-optimal rules; the distributions cross, so no comparison fully ranks them.}
		\label{Fig_MTE}
	\end{figure}
	
	\subsection{Optimal Rules and Exact Computation}\label{sec:emp_rules}
	
	Under the risk-neutral criterion, the budget-constrained problem is solved by ranking individuals according to their estimated gains, $\widehat{\Delta}_i$, and treating the top half. Under $\operatorname{AV@R}_{\beta}$, the optimal allocation depends on the entire welfare distribution rather than on treatment gains alone. Tractability is restored by the Rockafellar--Uryasev representation (Example~\ref{Example1}),
	\[
	\operatorname{AV@R}_{\beta}\bigl(\mathcal{W}(f,Z)\bigr)
	=
	\max_{\lambda\in\mathbb{R}}
	\left\{
	\lambda
	+
	\beta^{-1}
	\mathbb{E}_n
	\bigl[
	\mathcal{W}(f,Z)-\lambda
	\bigr]_-
	\right\}.
	\]
	
	Writing individual welfare as
	\[
	\mathcal{W}_i(\pi)
	=
	\hat\mu_{0,i}
	+
	\pi_i\widehat{\Delta}_i,
	\]
	the objective is separable across individuals for fixed $\lambda$. Hence, for each $\lambda$, the optimal allocation is obtained by treating the individuals with the largest positive treatment contributions, subject to the budget constraint. The remaining optimization is therefore one-dimensional in $\lambda$. Algorithm \ref{alg_AVAR_allocation} describes the computational procedure for obtaining the $\mathrm{AV@R}_{\beta}$-optimal treatment allocation.
	\begin{algorithm}
		\caption{Computation of the $\mathrm{AV@R}_{\beta}$-Optimal Budget-Constrained Allocation}
		\label{alg_AVAR_allocation}
		\begin{algorithmic}[1]
			
			\State Compute estimated untreated welfare and treatment gains:
			\[
			\hat{\mu}_{0,i}, \qquad \widehat{\Delta}_i,
			\]
			so that
			\[
			\mathcal{W}_i(\pi)=\hat{\mu}_{0,i}+\pi_i\widehat{\Delta}_i .
			\]
			
			\State For a candidate threshold $\lambda\in\mathbb{R}$, compute the marginal contribution of treating individual $i$:
			\[
			g_i(\lambda)
			=
			\beta^{-1}
			\left(
			[\hat{\mu}_{0,i}+\widehat{\Delta}_i-\lambda]_{-}
			-
			[\hat{\mu}_{0,i}-\lambda]_{-}
			\right).
			\]
			
			\State Rank individuals according to $g_i(\lambda)$ and assign treatment to the individuals with the largest positive values of $g_i(\lambda)$, subject to the budget constraint.
			
			\State Denote the resulting allocation $\pi_\lambda$, define $f_{\lambda} = \pi_\lambda\hat\Delta_i$, and evaluate the Rockafellar--Uryasev objective:
			\[
			\lambda+
			\beta^{-1}
			\mathbb{E}_{n}
			\left[
			\mathcal{W}(f_{\lambda},Z)-\lambda
			\right]_{-}.
			\]
			
			\State Search over $\lambda$:
			\[
			\hat{\lambda}
			\in
			\arg\max_{\lambda\in\mathbb{R}}
			\left\{
			\lambda+
			\beta^{-1}
			\mathbb{E}_{n}
			\left[
			\mathcal{W}(f_{\lambda},Z)-\lambda
			\right]_{-}
			\right\}.
			\]
			
			\State Return the allocation $\pi_{\hat{\lambda}}$.
			
		\end{algorithmic}
	\end{algorithm}
	
	\subsection{Results}
	
	The allocation rule that accounts for both endogeneity and risk aversion differs substantially from rules that disregard either or both. Table \ref{Tab_Rules} reports welfare estimates (under endogeneity) and treatment rates for the expected-welfare criterion and several $\text{AV@R}_\beta$ criteria. The mechanism underlying risk-averse allocation is laid bare: expected welfare is sacrificed to raise left-tail welfare, as the right panel of Figure \ref{Fig_MTE} illustrates directly---the $\text{AV@R}_{0.25}$-optimal rule places less mass in the lower tail of the welfare distribution than the risk-neutral rule, at the cost of less mass in the upper tail. Notably, the risk-averse planner reallocates treatment toward Black individuals despite no equity weights, group-specific objectives, or fairness constraints appearing in the planner's problem. This prioritization is an endogenous consequence of applying a statistical risk-averse criterion to a population in which this subgroup occupies the lower tail of the welfare distribution: concern for the worst-off becomes, in this population, concern for a recognizable demographic group.
	
	\begin{table}[t]
		\centering
		\caption{Optimal treatment rules under budget constraint $\Pi_{0.5}$}
		\label{Tab_Rules}
		\begin{tabular}{lcccccc}
			\hline\hline
			& \multicolumn{2}{c}{Welfare} & \multicolumn{4}{c}{Treatment rate by subgroup} \\
			\cmidrule(lr){2-3}\cmidrule(lr){4-7}
			Criterion & $\mathbb{E}[\mathcal{W}]$ & $\mathrm{AV@R}_{0.25}$
			& Black & Non-Black & South & Non-South \\
			\hline
			Risk-neutral            & 6.888 & 6.277 & 0.832 & 0.399 & 0.560 & 0.459 \\
			$\mathrm{AV@R}_{0.50}$  & 6.820 & 6.445 & 0.994 & 0.349 & 0.805 & 0.294 \\
			$\mathrm{AV@R}_{0.25}$  & 6.812 & 6.448 & 0.999 & 0.348 & 0.788 & 0.305 \\
			$\mathrm{AV@R}_{0.10}$  & 6.804 & 6.445 & 1.000 & 0.348 & 0.806 & 0.293 \\
			\hline\hline
		\end{tabular}
	\end{table}

	To isolate the cost of ignoring endogenous selection, we construct a naive planner who assumes unconfoundedness: treatment gains are estimated via the naive doubly-robust (AIPW) score of \citet{AtheyWager2021}, $$\hat\Gamma_i = \hat\mu_1(X_i) - \hat\mu_0(X_i) + \frac{D_i(Y_i-\hat\mu_1(X_i))}{\hat{e}(X_i)} - \frac{(1-D_i)(Y_i-\hat\mu_0(X_i))}{(1-\hat{e}(X_i))},$$ using the same estimated propensity score as the MTE stage---including $\mathrm{nearc4}$ as a predictor of $\hat{e}(Z)$, but not as an instrument. The naive planner thus sees college proximity only as a predictor of enrollment, not as a source of exogenous variation in gains. Feeding these estimates through the same budget-constrained $\text{AV@R}_{0.25}$ rule reveals a sharp asymmetry: naive and MTE-based baseline welfare are highly correlated ($\text{corr}(\hat\mu_0^{naive}, \hat\mu_0) = 0.878$), but the corresponding gain estimates are essentially unrelated ($\text{corr}(\hat\Gamma, \hat\Delta) = 0.003$). The two rules therefore agree on who is worst off but not on who benefits most: the resulting allocations disagree on 25.8\% of individuals, concentrated precisely where the gain estimates diverge most---the mean gap between naive and MTE-based gains is 0.926 among those treated differently, versus 0.665 among those treated the same. Ignoring endogenous selection does not distort the planner's assessment of disadvantage; it corrupts the assessment of who stands to gain---the object a risk-averse planner most needs to get right when reallocating treatment toward the lower tail.
	
	We close this section by pointing the reader to online Appendix B where we study the role of risk aversion using 
	the Job Training Partnership Act (JTPA) experimental data analyzed by
	\citet{Sasaki_Ura_2024}. Our results show that in this case, risk aversion has no bite.  Together the two applications bracket the empirical content
	of the theory, showing not only when a coherent risk measure reorders
	the planner's assignment but also when it provably does not. In addition, in the online Appendix D we analyze our results in a simulation study.

	\section{Conclusion and Extensions}\label{s7}
	
	This paper develops a framework for risk-averse treatment allocation under
	endogenous selection. It combines the marginal treatment effect framework with
	coherent risk measures, so that the planner values the distribution of
	welfare, not only its mean. We characterize the planner's problem through
	welfare bounds, robust dual representations, and Kusuoka representations, and
	establish finite-sample regret guarantees for empirical risk-averse policy
	learning. Stronger protection against adverse welfare outcomes increases the
	statistical complexity of learning an optimal treatment rule.
	
	In the \citet{Card1995} college-proximity data, risk aversion substantially
	changes treatment allocation even with the identification strategy and policy
	class held fixed. Risk-averse policies reallocate treatment toward
	lower-welfare individuals without explicit equity weights or group-specific
	objectives. Simulations confirm the theory's finite-sample predictions and
	show why endogenous selection must be accounted for.
	
	Several extensions remain. Localized complexity measures or low-noise
	conditions may yield sharper regret guarantees. Stronger moment restrictions
	or more information on the distribution of potential outcomes may tighten the
	welfare bounds. The regret analysis could also incorporate first-stage
	estimation error in the MTE components. Multi-valued, continuous, or dynamic
	treatments would require new identification and learning results,
	particularly to accommodate time-consistent risk preferences. Finally,
	inference on optimal risk-averse welfare and treatment rules remains an
	important open problem.
	\appendix
	\section{Proofs}\label{app:proofs}
	
	\subsection*{Proof of Theorem~\ref{Risk_Bounds}}
	\textit{Upper bound}: For any concave, monotone risk
	functional $\rho$, Jensen's inequality applied to
	conditional expectation gives
	$\rho(\mathbb{E}[Y(\pi)\mid Z])\geq\rho(Y(\pi))$
	(Corollary~6.52 of \cite{Shapiro_et_al_2013}): conditioning
	raises the risk-averse welfare, so $\rho(\mathcal{W}(f,Z))$
	provides an upper bound on $\rho(Y(\pi))$. Taking the maximum
	over $f\in\mathcal{F}$ gives the upper bound
	in~(\ref{Bounds_Welfare}).
	\textit{Lower bound}: from $\mathcal{W}(f,Z)-K_0\leq Y(\pi)$,
	monotonicity~(A1) gives
	$\rho(\mathcal{W}(f,Z)-K_0)\leq\rho(Y(\pi))$. Translation
	invariance~(A2) gives
	$\rho(\mathcal{W}(f,Z))-K_0\leq\rho(Y(\pi))$. Taking the max
	over $f$ completes the proof. \eproof
	
	\subsection*{Proof of Proposition~\ref{CVAR_cor}}
	Direct application of Theorem~\ref{Risk_Bounds} with
	$\rho=\mathrm{AV@R}_\beta$ and using the fact that  $\rho=\mathrm{AV@R}_\beta$  is a coherent risk measure. \eproof
	
	\subsection*{Proof of Proposition~\ref{Risk_measure_robust}}
	Since $\rho$ is concave and upper semicontinuous, the
	Fenchel-Moreau theorem
	\citep[Thm.~5]{Rockafellar_1974},
	\citep[Thm.~4.4.2]{AubinEkeland2006applied} yields the
	concave conjugate representation. \eproof
	
	\subsection*{Proof of Proposition~\ref{Prop_Coherent_Dual}}
	
	Since $\rho$ is coherent, \citet[Theorem~6.4]{Shapiro_et_al_2013}
	implies that $\rho(Z)=\inf_{\xi\in\mathcal{D}} \mathbb{E}_{\xi}[Z], $ where
	$
	\mathcal{D}
	=
	\left\{
	\xi\in\mathcal{Z}^*:
	\langle\xi,Z\rangle\le\rho(Z)\ \forall Z\in\mathcal{Z},
	\;
	\langle\xi,\mathbf1\rangle=1,
	\;
	\xi\ge0
	\right\}.
	$
	Translation invariance implies
	$\langle\xi,\mathbf1\rangle=1$, while monotonicity implies
	$\xi\ge0$. Moreover, $\mathcal{D}$ is nonempty by the dual
	representation theorem for coherent risk measures, convex by
	construction, and weak$^*$ compact as a weak$^*$-closed subset of
	the unit ball of $\mathcal{Z}^*$ by the Banach--Alaoglu theorem.\\
	For
	$\rho=\operatorname{AV@R}_{\beta}$,
	the characterization of $\mathcal{D}$ follows from
	\citet[Example~6.5]{Shapiro_et_al_2013}. \eproof

	\subsection*{Proof of Lemma~\ref{Robust_Social_Welfare}}
	Representation~\eqref{rho_varphi_representation} follows directly from the variational representation of concave risk measures induced by $\varphi$-divergences. Specifically, \citet[Theorem~4.2]{Ben-Tal_Teboulle_2007} identify the concave conjugate $-\rho^*$ with the $\varphi$-divergence penalty $I_\varphi(\mathbb{Q},\mathbb{P})$. Substituting this characterization into the dual representation of $\rho$ immediately yields~\eqref{rho_varphi_representation}. \eproof

	\subsection*{Proof of Proposition~\ref{Cor_program_varphi}}
	
	We consider the risk function $\rho_\varphi(\mathcal{W}(f,Z)) = \inf_{Q}\{\mathbb{E}_Q[\mathcal{W}(f,Z)] + I_\varphi(Q,P)\}$.
	Theorem 4.2 of \citet{Ben-Tal_Teboulle_2007}, stated there for a loss variable $Z=-W$
	and translated into the reward convention used here, gives
	\[
	\rho_\varphi(\mathcal{W}(f,Z)) = \sup_{\eta\in\mathbb{R}}
	\big\{\eta - \mathbb{E}_P[\varphi^{*}(\eta-\mathcal{W}(f,Z))]\big\},
	\]
	which yields (\ref{program_varphi2}). The second result follows by taking
	the supremum over $f\in\mathcal{F}$.
	\eproof
	
	\subsection*{Proof of Corollary \ref{cor:excess_welfare}}
	
	By Proposition~\ref{Cor_program_varphi}, $\rho_\varphi(\mathcal{W}(f,Z))=\sup_{\eta\in\mathbb{R}}\{\eta-\EE[\varphi^*(\eta-\mathcal{W}(f,Z))]\}$. Under Assumptions~\ref{Assump_MTE}, \ref{Assumption2_OCE_Bounded} and~\ref{Ass_Y0_Bounded}, $\mathcal{W}(f,Z)$ is bounded, so $\EE[\mathcal{W}(f,Z)]$ is finite. Fix $\eta\in\mathbb{R}$. Using $\eta=\EE[\eta]$ and linearity of the expectation,
	\[
	\eta-\EE\bigl[\varphi^*(\eta-\mathcal{W}(f,Z))\bigr]-\EE[\mathcal{W}(f,Z)]
	=\EE\bigl[\eta-\mathcal{W}(f,Z)-\varphi^*(\eta-\mathcal{W}(f,Z))\bigr]
	=\EE\bigl[\phi(\eta-\mathcal{W}(f,Z))\bigr],
	\]
	where the last equality applies $\phi(t)=t-\varphi^*(t)$ at $t=\eta-\mathcal{W}(f,Z)$. Since $\EE[\mathcal{W}(f,Z)]$ does not depend on $\eta$, taking the supremum over $\eta\in\mathbb{R}$ gives
	\[
	\rho_\varphi(\mathcal{W}(f,Z))-\EE[\mathcal{W}(f,Z)]=\sup_{\eta\in\mathbb{R}}\EE\bigl[\phi(\eta-\mathcal{W}(f,Z))\bigr]=\psi(\mathcal{W}(f,Z)),
	\]
	which is the claimed representation. \eproof

	\subsection*{Proof of Theorem~\ref{Kusuoka_Welfare}}
	Part~(i) follows from Theorem~4.62 of
	\cite{FollmerSchied_2025}. Part~(ii) is Kusuoka's
	representation \citep{Shapiro_Kusuoka2013}. \eproof
	
	\subsection*{\textcolor{black}{Proof of Proposition~\ref{Prop_SOSD}}}
	(1)$\Rightarrow$(2). Suppose $\mathcal{W}(f_1,Z)\succeq_2\mathcal{W}(f_2,Z)$. By the characterization of second-order stochastic dominance in terms of Average Value-at-Risk \citep[Theorem~2.57]{FollmerSchied_2025},
	\begin{equation}\label{eq:SOSD_AVaR}
		AV@R_\alpha\big(\mathcal{W}(f_1,Z)\big)\;\geq\;AV@R_\alpha\big(\mathcal{W}(f_2,Z)\big)
		\qquad\text{for every }\alpha\in(0,1].
	\end{equation}
	Let $\rho$ be any law-invariant, concave, and monotone risk measure. Under Assumptions~\ref{Assump_MTE}, ~\ref{Assumption2_OCE_Bounded} and ~\ref{Ass_Y0_Bounded}, $\mathcal{W}(f,Z)\in L^\infty$ for every $f\in\mathcal{F}$, so Theorem~\ref{Kusuoka_Welfare}(i) applies. Using this representation yields:
	\begin{multline*}
		\rho\big(\mathcal{W}(f_1,Z)\big)
		=\inf_{\mu\in\mathcal{M}}\left\{\int_0^1 AV@R_\alpha\big(\mathcal{W}(f_1,Z)\big)\,\mu(d\alpha)+c(\mu)\right\} \geq \\
		\inf_{\mu\in\mathcal{M}}\left\{\int_0^1 AV@R_\alpha\big(\mathcal{W}(f_2,Z)\big)\,\mu(d\alpha)+c(\mu)\right\} = \rho\big(\mathcal{W}(f_2,Z)\big)
	\end{multline*}
	
	(2)$\Rightarrow$(1). For every $\alpha\in(0,1]$, $AV@R_\alpha$ is concave, monotone, and law-invariant. Applying (2) to $\rho=AV@R_\alpha$ for each $\alpha\in(0,1]$ gives \eqref{eq:SOSD_AVaR}. By \citet[Theorem~2.57]{FollmerSchied_2025} this is equivalent to $\mathcal{W}(f_1,Z)\succeq_2\mathcal{W}(f_2,Z)$. \eproof
	\textcolor{black}{
		\subsection*{Proof of Lemma~\ref{Lem_OCE_Range}}
		Write $h(\eta)\triangleq\eta-\EE[\varphi^*(\eta-Z)]$. Since $\varphi^*$ is convex, $h$ is concave, with $h'(\eta)=1-\EE[\varphi^{*\prime}(\eta-Z)]$, where differentiation and expectation may be interchanged by dominated convergence, since $Z$ is bounded and $\varphi^{*\prime}$ is bounded on the relevant range. Because $\varphi(1)=0=\inf_t\varphi(t)$, $t=1$ minimizes $\varphi$ and $0\in\partial\varphi(1)$; Fenchel duality ($s\in\partial\varphi(t)\iff t\in\partial\varphi^*(s)$) then gives $\varphi^{*\prime}(0)=1$. Since $\varphi^{*\prime}$ is nondecreasing, for $\eta>b$ every realization of $Z$ satisfies $\eta-Z>0$, so $\varphi^{*\prime}(\eta-Z)\geq\varphi^{*\prime}(0)=1$ a.s., hence $h'(\eta)\leq0$ on $(b,\infty)$; symmetrically, for $\eta<a$ we have $\eta-Z<0$ a.s., so $\varphi^{*\prime}(\eta-Z)\leq\varphi^{*\prime}(0)=1$ a.s.\ and $h'(\eta)\geq0$ on $(-\infty,a)$. Thus $h$ is nonincreasing on $[b,\infty)$ and nondecreasing on $(-\infty,a]$, and its supremum over $\mathbb{R}$ equals its supremum over $[a,b]$. \eproof}
	
	\subsection*{Proof of Lemma~\ref{Lem_Contraction}}
	Every $g\in\mathcal{G}$ satisfies $g(Z_i)=\eta-\mathcal{W}(f,Z_i)\in[a-b,b-a]=[-K,K]$, so $\varphi^*$ restricted to the relevant range is $L_\varphi$-Lipschitz and satisfies $\varphi^*(0)=-\inf_t\varphi(t)=0$, so the contraction principle \citep[Lemma~26.9]{Shalev-Shwartz_Ben-David_2014} gives the first inequality. For the second, write $g(Z_i)=\eta-\EE(Y_0| Z_i)-f(Z_i)$, so that
	\begin{align*}
		\mathfrak{R}_n(\mathcal{G}(Z^n))
		&=\EE_{\sigma^n}\left[\sup_{f\in\mathcal{F},\,\eta\in[a,b]}\frac{1}{n}\sum_{i=1}^n\sigma_i\bigl(\eta-\EE(Y_0\mid Z_i)-f(Z_i)\bigr)\right]\\
		&\leq\EE_{\sigma^n}\left[\sup_{f\in\mathcal{F}}\frac{1}{n}\sum_{i=1}^n\sigma_i f(Z_i)\right]
		+\EE_{\sigma^n}\left[\sup_{\eta\in[a,b]}\eta\cdot\frac{1}{n}\sum_{i=1}^n\sigma_i\right]
		+\left|\EE_{\sigma^n}\frac{1}{n}\sum_{i=1}^n\sigma_i\EE(Y_0\mid Z_i)\right|\\
		&=\mathfrak{R}_n(\mathcal{F},Z^n)+\max(|a|,|b|)\cdot\EE_{\sigma^n}\left|\frac{1}{n}\sum_{i=1}^n\sigma_i\right|+0\\
		&\leq\mathfrak{R}_n(\mathcal{F},Z^n)+\frac{K}{2\sqrt{n}},
	\end{align*}
	where the last line uses $\max(|a|,|b|)=\frac{K}{2}$ together with $\EE_{\sigma^n}|n^{-1}\sum_i\sigma_i|\leq(\EE_{\sigma^n}(n^{-1}\sum_i\sigma_i)^2)^{1/2}=n^{-1/2}$ (Jensen), and the third term vanishes because $\EE(Y_0\mid Z_i)$ does not depend on $f$ or $\eta$, so it may be taken outside the supremum, whereupon $\EE_{\sigma^n}[\sigma_i]=0$. \eproof
	
	\subsection*{Proof of Lemma~\ref{Lem_Symmetrization}}
	Write $A_f(\eta)\triangleq\eta-\EE_{\mathbb{P}}[\varphi^*(\eta-\mathcal{W}(f,Z))]$ and $B_f(\eta)\triangleq\eta-\EE_{\mathbb{P}_n}[\varphi^*(\eta-\mathcal{W}(f,Z))]$. By Proposition~\ref{Cor_program_varphi} and Lemma~\ref{Lem_OCE_Range}, for every $f\in\mathcal{F}$,
	$$
	\rho_\varphi(\mathcal{W}(f,Z))=\sup_{\eta\in[a,b]}A_f(\eta),\qquad
	\rho_\varphi(\mathcal{W}_n(f,Z))=\sup_{\eta\in[a,b]}B_f(\eta).
	$$
	Since $A_f=(A_f-B_f)+B_f$, subadditivity of the supremum gives
	$$
	\sup_{\eta\in[a,b]}A_f(\eta)\;\leq\;\sup_{\eta\in[a,b]}\{A_f(\eta)-B_f(\eta)\}+\sup_{\eta\in[a,b]}B_f(\eta),
	$$
	and since the outer $\eta$ cancels in $A_f-B_f$, rearranging yields
	$$
	\rho_\varphi(\mathcal{W}(f,Z))-\rho_\varphi(\mathcal{W}_n(f,Z))
	\leq \sup_{\eta\in[a,b]}\bigl\{\EE_{\mathbb{P}_n}[\varphi^*(\eta-\mathcal{W}(f,Z))]-\EE_{\mathbb{P}}[\varphi^*(\eta-\mathcal{W}(f,Z))]\bigr\}.
	$$
	Taking the supremum over $f\in\mathcal{F}$, and recalling that $\mathcal{G}\triangleq\{\eta-\mathcal{W}(f,\cdot):f\in\mathcal{F},\;\eta\in[a,b]\}$ yields
	$$
	\sup_{f\in\mathcal{F}}\bigl(\rho_\varphi(\mathcal{W}(f,Z))-\rho_\varphi(\mathcal{W}_n(f,Z))\bigr)
	\leq\sup_{g\in\mathcal{G}}\bigl(\EE_{\mathbb{P}_n}[\varphi^*\circ g]-\EE_{\mathbb{P}}[\varphi^*\circ g]\bigr).
	$$
	Unlike $\rho_\varphi(\mathcal{W}_n(f,Z))$ itself, the right side is a genuine linear-in-the-sample empirical process ($\EE_{\mathbb{P}_n}[\varphi^*\circ g]$ is a sample mean), so the standard symmetrization argument applies without difficulty. Introducing a ghost sample $\{Z_i'\}_{i=1}^n$ i.i.d.\ from $\mathbb{P}$, using $\EE_{\mathbb{P}}[\varphi^*\circ g]=\EE_{Z'}\bigl[\tfrac{1}{n}\sum_i\varphi^*(g(Z_i'))\bigr]$, Jensen's inequality, and Rademacher signs $\{\sigma_i\}$ (since $(\sigma_iZ_i',\sigma_iZ_i)\overset{d}{=}(Z_i',Z_i)$),
	$$
	\EE\left[\sup_{g\in\mathcal{G}}\bigl(\EE_{\mathbb{P}_n}[\varphi^*\circ g]-\EE_{\mathbb{P}}[\varphi^*\circ g]\bigr)\right]
	\leq 2\,\mathfrak{R}_n(\{\varphi^*\circ g:g\in\mathcal{G}\})
	\leq 2L_\varphi\left(\EE[\mathfrak{R}_n(\mathcal{F},Z^n)]+\frac{K}{2\sqrt{n}}\right),
	$$
	where the last step is Lemma~\ref{Lem_Contraction}. Combining the displays gives~(\ref{Symmetrization_ineq}). \eproof

	\subsection*{Proof of Proposition~\ref{Prop_Kusuoka_Symmetrization}}
	By Theorem~\ref{Kusuoka_Welfare}(ii), $\rho(\cdot)=\inf_{\mu\in\mathcal{M}_\rho}\int_0^1 AV@R_\alpha(\cdot)\,\mu(d\alpha)$. Fix $f\in\mathcal{F}$ and let $\mu^\star_n(f)\in\mathcal{M}_\rho$ attain the infimum defining $\rho(\mathcal{W}_n(f,Z))$.\footnote{If the infimum is not attained, replace $\mu^\star_n(f)$ by an $\varepsilon$-minimizer; the same bound follows on letting $\varepsilon\downarrow 0$.} Evaluating $\rho(\mathcal{W}(f,Z))$ at the same $\mu^\star_n(f)$ and using that the infimum is a lower bound,
	\begin{align*}
		\rho(\mathcal{W}(f,Z))-\rho(\mathcal{W}_n(f,Z))
		&\leq\int_0^1\bigl(AV@R_\alpha(\mathcal{W}(f,Z))-AV@R_\alpha(\mathcal{W}_n(f,Z))\bigr)\,\mu^\star_n(f)(d\alpha)\\
		&\leq\sup_{\alpha\in[\alpha_{\min},1]}\bigl(AV@R_\alpha(\mathcal{W}(f,Z))-AV@R_\alpha(\mathcal{W}_n(f,Z))\bigr),
	\end{align*}
	where the second inequality uses Assumption~\ref{Ass_KusuokaTail} to restrict $\operatorname{supp}(\mu^\star_n(f))$ to $[\alpha_{\min},1]$, together with the fact that an integral against a probability measure never exceeds the supremum of the integrand.
	
	Each $AV@R_\alpha$ is a $\varphi$-divergence risk measure with conjugate $\varphi^*_{AV@R_\alpha}(s)=s_+/\alpha$. Taking $\sup_f$ and the expectation over the sample, the argument of Lemma~\ref{Lem_Symmetrization} applies with $\mathcal{G}$ replaced by $\widetilde{\mathcal{G}}\triangleq\{\varphi^*_{AV@R_\alpha}\circ g:g\in\mathcal{G},\,\alpha\in[\alpha_{\min},1]\},$
	indexed by $(f,\eta,\alpha)$ rather than $(f,\eta)$, since the optimized certainty equivalent representation, the range restriction of Lemma~\ref{Lem_OCE_Range}, and the subadditivity step all hold pointwise in $\alpha$. Symmetrization therefore gives
	\begin{equation}\label{eq:Kusuoka_symm_step}
		\EE\Bigl[\sup_{f\in\mathcal{F}}\bigl\{\rho(\mathcal{W}(f,Z))-\rho(\mathcal{W}_n(f,Z))\bigr\}\Bigr]\;\leq\;2\,\EE\bigl[\mathfrak{R}_n(\widetilde{\mathcal{G}};Z^n)\bigr].
	\end{equation}
	
	It remains to bound $\mathfrak{R}_n(\widetilde{\mathcal{G}};Z^n)$. Writing $\kappa\triangleq1/\alpha\in[1,1/\alpha_{\min}]$ and $\mathcal{H}\triangleq\bigl\{(\eta-\mathcal{W}(f,\cdot))_+:f\in\mathcal{F},\,\eta\in[a,b]\bigr\},$
	we have $\widetilde{\mathcal{G}}=\{\kappa h:h\in\mathcal{H},\,\kappa\in[1,1/\alpha_{\min}]\}$, with $h(Z_i)\in[0,K]$ for every $h\in\mathcal{H}$. Because $\kappa$ varies freely rather than being fixed, the contraction principle does not by itself bound $\mathfrak{R}_n(\widetilde{\mathcal{G}};Z^n)$ by $\alpha_{\min}^{-1}\mathfrak{R}_n(\mathcal{H};Z^n)$. For fixed $\sigma^n$ write $S_h\triangleq n^{-1}\sum_{i=1}^n\sigma_i h(Z_i)$. Since $\kappa\mapsto\kappa S_h$ is linear and $\kappa\geq1>0$, the pointwise-optimal $\kappa$ given $h$ is $1/\alpha_{\min}$ when $S_h\geq0$ and $1$ otherwise, so
	\[
	\sup_{h\in\mathcal{H},\,\kappa\in[1,1/\alpha_{\min}]}\kappa S_h\;\leq\;\frac{1}{\alpha_{\min}}\Bigl(\sup_{h\in\mathcal{H}}S_h\Bigr)^{+}.
	\]
	Applying the identity $x^{+}=x+x^{-}$ at $x=\sup_h S_h$, where $x^{+}\triangleq\max\{x,0\}$ and $x^{-}\triangleq\max\{-x,0\}$, and taking expectations over $\sigma^n$,
	\begin{equation}\label{eq:Kusuoka_residual_split}
		\mathfrak{R}_n(\widetilde{\mathcal{G}};Z^n)\;\leq\;\frac{1}{\alpha_{\min}}\Bigl(\mathfrak{R}_n(\mathcal{H};Z^n)+\EE_{\sigma^n}\bigl[(\sup_h S_h)^{-}\bigr]\Bigr),
	\end{equation}
	where $\mathfrak{R}_n(\mathcal{H};Z^n)\geq0$ by Jensen's inequality. The second term in \eqref{eq:Kusuoka_residual_split} is the residual generated by the free index $\alpha$; it is not proportional to $\mathfrak{R}_n(\mathcal{H};Z^n)$ and so contributes additively.
	
	To bound the residual, note that flipping a single $\sigma_i$ changes $\sup_h S_h$ by at most $2K/n$, since $h(Z_i)\in[0,K]$. By McDiarmid's inequality \citep[Theorem~6.2]{BoucheronLugosiMassart2013}, $Y\triangleq\sup_h S_h-\mathfrak{R}_n(\mathcal{H};Z^n)$ satisfies $\mathbb{P}(|Y|>t)\leq 2e^{-t^2/(2\sigma_m^2)}$ with $\sigma_m^2\triangleq K^2/n$. Integrating this tail yields $\EE|Y|\leq \sigma_m\sqrt{2\pi}$.
	Since $\mathfrak{R}_n(\mathcal{H};Z^n)\geq0$ we have $(\sup_h S_h)^{-}\leq Y^{-}$, and $\EE[Y^{-}]=\tfrac12\EE|Y|$ because $Y$ is mean zero, so
	\[
	\EE_{\sigma^n}\bigl[(\sup_h S_h)^{-}\bigr]\;\leq\;\frac{\sigma_mn\sqrt{2\pi}}{2}\;=\;\frac{K}{\sqrt{n}}\sqrt{\frac{\pi}{2}}.
	\]
	Finally, $t\mapsto t^{+}$ is $1$-Lipschitz with $0^{+}=0$, so the contraction principle \citep[Lemma~26.9]{Shalev-Shwartz_Ben-David_2014} gives $\mathfrak{R}_n(\mathcal{H};Z^n)\leq\mathfrak{R}_n(\mathcal{G};Z^n)$, and Lemma~\ref{Lem_Contraction} gives $\mathfrak{R}_n(\mathcal{G};Z^n)\leq\mathfrak{R}_n(\mathcal{F};Z^n)+K/(2\sqrt{n})$. Substituting into \eqref{eq:Kusuoka_residual_split} and then \eqref{eq:Kusuoka_symm_step} yields \eqref{Kusuoka_Symmetrization_ineq} with
	\[
	\epsilon_n=\frac{2}{\alpha_{\min}}\cdot\frac{K}{\sqrt{n}}\sqrt{\frac{\pi}{2}}=\frac{K}{\alpha_{\min}}\sqrt{\frac{2\pi}{n}}=O(n^{-1/2}).
	\]
	\eproof
	
	\subsection*{Proof of Proposition~\ref{Prop_UC}}
	We prove~(i); (ii) is identical with $L_\varphi\to1/\alpha_{\min}$, invoking Proposition~\ref{Prop_Kusuoka_Symmetrization} instead of Lemma~\ref{Lem_Symmetrization} throughout and carrying the additional $\epsilon_n$ slack. Define $H(Z^n)\triangleq\sup_{f\in\mathcal{F}}|
	\rho_\varphi(\mathcal{W}(f,Z))-\rho_\varphi(\mathcal{W}_n(f,Z))|$.
	
	\smallskip
	\noindent\textit{Step~1: Bounding $\mathbb{E}[H(Z^n)]$.}
	Since $H(Z^n)\leq\sup_{f\in\mathcal{F}}
	(\rho_\varphi(\mathcal{W}(f,Z))-\rho_\varphi(\mathcal{W}_n(f,Z)))+
	\sup_{f\in\mathcal{F}}(\rho_\varphi(\mathcal{W}_n(f,Z))-
	\rho_\varphi(\mathcal{W}(f,Z)))$, applying
	Lemma~\ref{Lem_Symmetrization} to each term (the second
	term is the symmetrized version with the roles of
	$\mathbb{P}$ and $\mathbb{P}_n$ exchanged, which has the
	same bound):
	\begin{equation}\label{Step1_UC}
		\mathbb{E}[H(Z^n)]
		\leq 2\cdot 2L_\varphi\!\left(\mathbb{E}[\mathfrak{R}_n(\mathcal{F},Z^n))]+\frac{K}{2\sqrt n}\right)
		= 4L_\varphi\,\mathbb{E}[\mathfrak{R}_n(\mathcal{F},Z^n)]+\frac{2L_\varphi K}{\sqrt n}.
	\end{equation}
	
	\noindent\textit{Step~2: Concentration via McDiarmid.}
	Replacing any single $Z_i$ by an independent copy $Z_i'$ changes $\rho_\varphi(\mathcal{W}_n(f,Z))=\sup_{\eta\in[a,b]}\{\eta-\tfrac1n\sum_j\varphi^*(\eta-\mathcal{W}(f,Z_j))\}$, for fixed $f$ and $\eta$, by at most $L_\varphi K/n$ (only the $i$-th summand changes, and $\varphi^*$ is $L_\varphi$-Lipschitz on the relevant range by Lemma~\ref{Lem_Contraction}); a supremum of functions each changing by at most $L_\varphi K/n$ itself changes by at most $L_\varphi K/n$, so $H(Z^n)$ changes by at most \textcolor{black}{$L_\varphi K/n$ (the population term is nonrandom, and both $|\cdot|$ and $\sup_f$ preserve the bound)}. By McDiarmid's inequality \citep[Theorem~6.2]{BoucheronLugosiMassart2013}:
	\begin{equation}\label{Step2_UC}
		\mathbb{P}[H(Z^n)-\mathbb{E}[H(Z^n)]\geq t]
		\leq\exp\!\left(-\frac{2nt^2}{L_\varphi^2K^2}\right).
	\end{equation}
	Setting the right side equal to $\delta$ and solving gives
	$t=L_\varphi K\sqrt{\log(1/\delta)/(2n)}$.
	Combining~(\ref{Step1_UC}) and~(\ref{Step2_UC}) then yields~(\ref{UC_bound})
	with probability at least $1-\delta$. \eproof
	
	\subsection*{Proof of Theorem~\ref{Thm_Regret}}
	Let $\hat{f}_n=\arg\max_{f\in\mathcal{F}}
	\rho(\mathcal{W}_n(f,Z))$. and Fix $\varepsilon>0$ and let $f^{*}_{\varepsilon}$ satisfy
	$\rho(\mathcal{W}(f^{*}_{\varepsilon},Z)) \ge
	\sup_{f\in\mathcal{F}}\rho(\mathcal{W}(f,Z)) - \varepsilon$. Decompose:
	\begin{align}
		\mathcal{R}_\rho(\hat{f}_n)
		&=\rho(\mathcal{W}(f^{*}_{\varepsilon},Z))-\rho(\mathcal{W}(\hat{f}_n,Z))
		\nonumber\\
		&=\underbrace{[\rho(\mathcal{W}(f^{*}_{\varepsilon},Z))
			-\rho(\mathcal{W}_n(f^{*}_{\varepsilon},Z))]}_{(A)}
		+\underbrace{[\rho(\mathcal{W}_n(f^{*}_{\varepsilon},Z))
			-\rho(\mathcal{W}_n(\hat{f}_n,Z))]}_{(B)}\nonumber\\
		&\quad+\underbrace{[\rho(\mathcal{W}_n(\hat{f}_n,Z))
			-\rho(\mathcal{W}(\hat{f}_n,Z))]}_{(C)}.
		\label{Regret_decomp}
	\end{align}
	Since $\hat{f}_n$ maximizes $\rho(\mathcal{W}_n(\cdot,Z))$
	over $\mathcal{F}$, term~$(B)\leq 0$. Dropping~$(B)$ and
	bounding each of $(A)$ and $(C)$ by
	$\sup_{f\in\mathcal{F}}|\rho(\mathcal{W}(f,Z))
	-\rho(\mathcal{W}_n(f,Z))|$ gives
	$\mathcal{R}_\rho(\hat{f}_n)\leq 2\sup_{f\in\mathcal{F}}
	|\rho(\mathcal{W}(f,Z))-\rho(\mathcal{W}_n(f,Z))|$.
	Applying Proposition~\ref{Prop_UC} to the right side and letting $\epsilon\rightarrow 0$
	gives~(\ref{Regret_bound}) and (\ref{Regret_bound_AVAR}). \eproof

	\bigskip
	\begin{center}
		{\Large\bfseries Online Supplemental Appendix}
	\end{center}
	\bigskip
	\noindent The remaining appendices may be skipped on a first reading. Appendix~\ref{app:jtpa} repeats the empirical analysis of Section~\ref{sec:empirical} on the JTPA data; Appendix~\ref{app:semidev} gives the full development of the mean--semideviation risk measure introduced in Section~\ref{s4}; Appendix~\ref{app:simulation} gives the full simulation study.
	
	\refstepcounter{section}
	
	\section*{Appendix B: The JTPA Application: When Risk Aversion Has No Bite}\label{app:jtpa}
	
	This appendix repeats the analysis of Section \ref{sec:empirical} of the main paper on
	the Job Training Partnership Act (JTPA) experimental data analyzed by
	\citet{Sasaki_Ura_2024}. The exercise serves two purposes. First, it
	engages the leading risk-neutral endogenous-treatment benchmark on its
	own dataset. Second---and this is the substantive finding---it
	delivers the opposite configuration from the Card application: on the
	JTPA data, the risk-neutral and risk-averse optimal rules essentially
	coincide. Together the two applications bracket the empirical content
	of the theory, showing not only when a coherent risk measure reorders
	the planner's assignment but also when it provably does not. A
	methodology that manufactured tail effects wherever it was pointed
	would be suspect; the JTPA results demonstrate that the criterion
	responds to a specific, diagnosable feature of the joint distribution
	of gains and baseline welfare, and reports its absence when absent.
	
	\subsection{Data and Estimation}\label{app:jtpa_setup}
	
	The sample comprises $n = 9{,}872$ adults from the National JTPA
	Study. The outcome $Y$ is log total earnings in the 30 months
	following random assignment, treatment $D$ is actual enrollment in
	JTPA services, and the instrument $Z_0$ is the randomized offer of
	services. Compliance is one-sided to a close approximation: $66.1\%$
	of those offered services enroll, against $1.5\%$ of those not
	offered. Covariates $X$ comprise gender, race, ethnicity, marital
	status, an indicator for a high-school diploma or GED, an indicator
	for having worked fewer than 13 weeks in the preceding year, AFDC
	receipt, and age-bracket indicators.
	
	Random assignment makes the instrument unconditionally valid but
	renders the propensity score nearly binary: $\hat{v}(Z_0 = 0) \approx
	0$, while $\hat{v}(Z_0 = 1)$ varies only modestly around $0.66$ with
	covariates. Nonparametric identification of the MTE over the full unit
	interval is therefore unavailable, and the local-IV regression is
	prone to numerical instability. Following the parametric
	extrapolation strategy of \citet{Sasaki_Ura_2024}, we estimate the
	outcome equation with a degree-3 polynomial in the propensity score,
	interacted with race and gender, and stabilize the near-collinear
	design with a Ridge penalty. Figure~\ref{Fig_JTPA_MTE} reports the
	estimated MTE; the level is anchored to the two-stage least squares
	estimated ($\hat{\beta}_{\mathrm{2SLS}} = 0.116$), consistent with the
	main text.
	
	The estimated gains are the first place the JTPA data diverge from
	Card. The mean individual gain is \emph{negative}
	($\bar{\widehat{\Delta}} = -0.17$ log points, s.d.\ $0.20$), and only
	$22.1\%$ of individuals have positive estimated gains ($39.1\%$ among
	Black individuals, $43.2\%$ among those without a high-school
	credential, $24.9\%$ among women). This is consistent with the modest
	experimental impacts documented in the JTPA literature
	\citep{BloomEtAl1997}: for most individuals, training is not estimated
	to raise earnings.
	
	\begin{figure}[t]
		\centering
		\includegraphics[width=\textwidth]{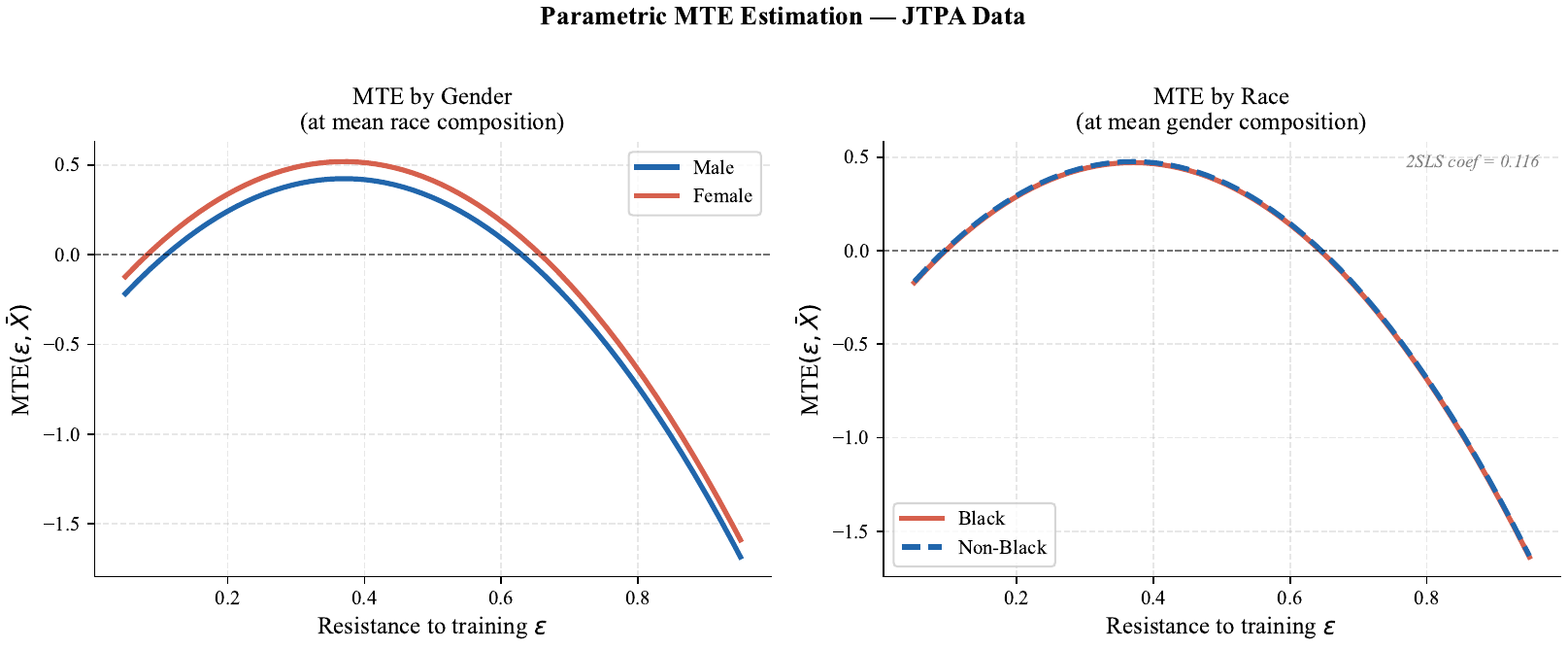}
		\caption{Parametric MTE estimates on the JTPA data, by gender and by
			race, with the level anchored to the 2SLS estimand. Identification
			over the interior of the unit interval rests on the parametric
			extrapolation discussed in the text.}
		\label{Fig_JTPA_MTE}
	\end{figure}
	
	\subsection{The Slack-Budget Regime: All Criteria Coincide}\label{app:jtpa_slack}
	
	Consider first the same $50\%$ budget constraint as the main text.
	Because only $22.1\%$ of individuals have positive estimated gains,
	the constraint is slack: every optimal rule treats exactly the
	positive-gain set and leaves budget unused. The four criteria---the
	mean and $\mathrm{AV@R}_{\beta}$ for $\beta \in \{0.50, 0.25,
	0.10\}$---produce \emph{identical} assignments, each treating $22.1\%$
	of the sample and attaining $\mathbb{E}[\mathcal{W}] = 9.079$ and
	$\mathrm{AV@R}_{0.25} = 8.557$.
	
	This exact coincidence is not a numerical accident but monotonicity
	made visible. Welfare is additively separable across individuals,
	$\mathcal{W}_i(\pi) = \mu_{0,i} + \pi_i \Delta_i$, so treating an
	individual with $\Delta_i > 0$ raises welfare pointwise and treating
	one with $\Delta_i < 0$ lowers it pointwise. Every coherent risk
	measure is monotone, and hence---whenever the budget exceeds the mass
	of positive-gain individuals---is maximized by treating precisely that
	set, exactly as the mean is. Risk aversion can only matter when the
	planner is forced to \emph{ration} among individuals who would
	benefit, or to trade one individual's welfare against another's; a
	slack budget forecloses both channels. The main text's $50\%$
	constraint binds on the Card data (where $96\%$ of gains are positive)
	and is slack here, which is itself informative about the two
	populations.
	
	\subsection{A Binding Budget: Economically Negligible Separation}\label{app:jtpa_binding}
	
	To give risk aversion its best chance, we tighten the budget to
	$10\%$, below the positive-gain share, so that rationing bites. The
	rules now differ, but barely. The risk-neutral and
	$\mathrm{AV@R}_{0.25}$-optimal rules disagree on $8.0\%$ of
	assignments; the risk-averse rule improves $\mathrm{AV@R}_{0.25}$ by
	$0.005$ log points ($8.5515 \to 8.5568$) at a mean-welfare cost of
	$0.004$---two orders of magnitude below the $0.171$ tail gain in the
	Card application, and economically negligible on any
	reading.\footnote{The $\mathrm{AV@R}_{0.25}$-optimal rule at the
		$10\%$ budget attains \emph{exactly} the slack-budget value of
		$8.5568$: so few individuals in the welfare tail have positive
		estimated gains that a $10\%$ budget already exhausts everything the
		risk-averse criterion can accomplish.} Subgroup treatment rates shift
	mildly toward lower-baseline groups, but the differences are within
	the flat region of the objective and we do not emphasize them.
	Figures~\ref{Fig_JTPA_CDF} and \ref{Fig_JTPA_AVaR} display the
	induced welfare distributions and the $\mathrm{AV@R}_{\beta}$
	profiles: in contrast to the crossing distributions seen for the Card
	data in the main text, the curves here are visually indistinguishable,
	with a maximal profile gap below $0.05$ log points.
	
	\begin{figure}[t]
		\centering
		\includegraphics[width=\textwidth]{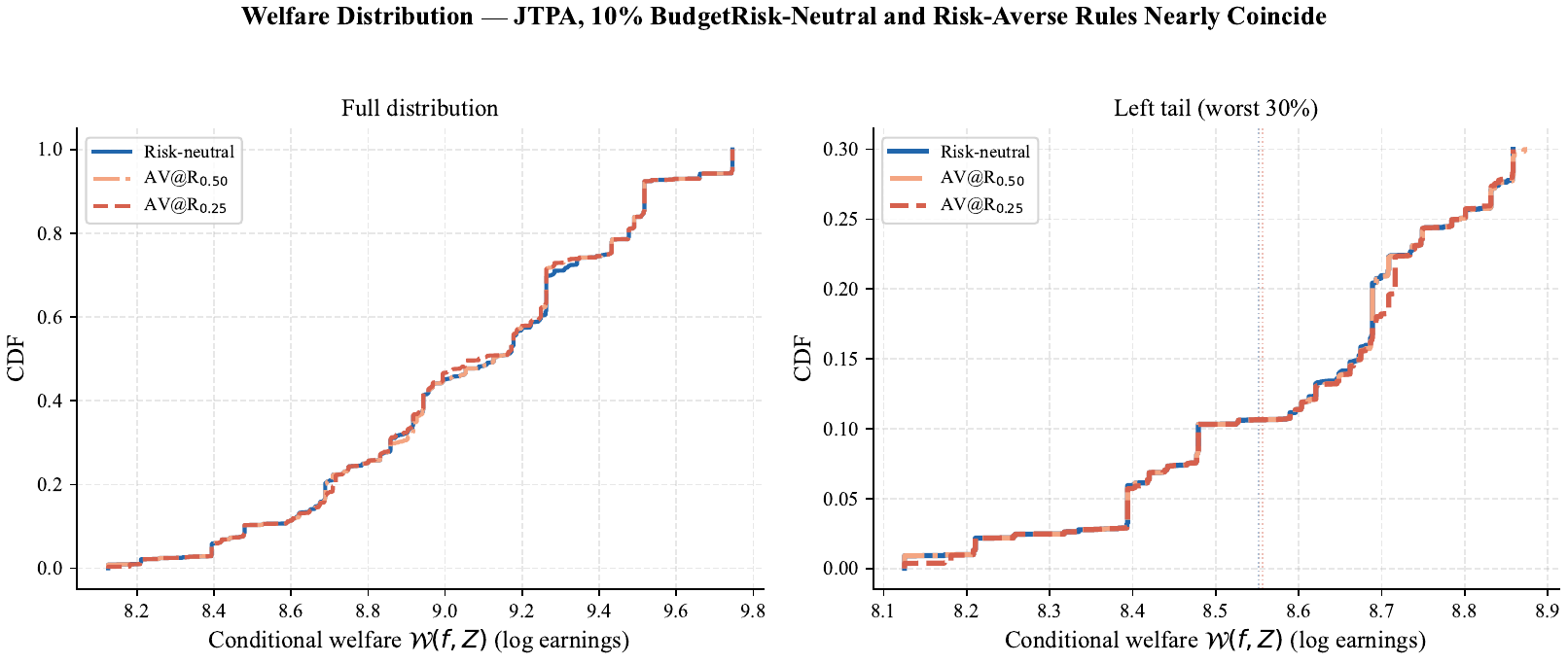}
		\caption{Welfare distributions induced by the optimal rules on the
			JTPA data under the binding $10\%$ budget. Unlike the Card
			application (main text), the distributions under
			risk-neutral and risk-averse rules nearly coincide.}
		\label{Fig_JTPA_CDF}
	\end{figure}
	
	\begin{figure}[t]
		\centering
		\includegraphics[width=\textwidth]{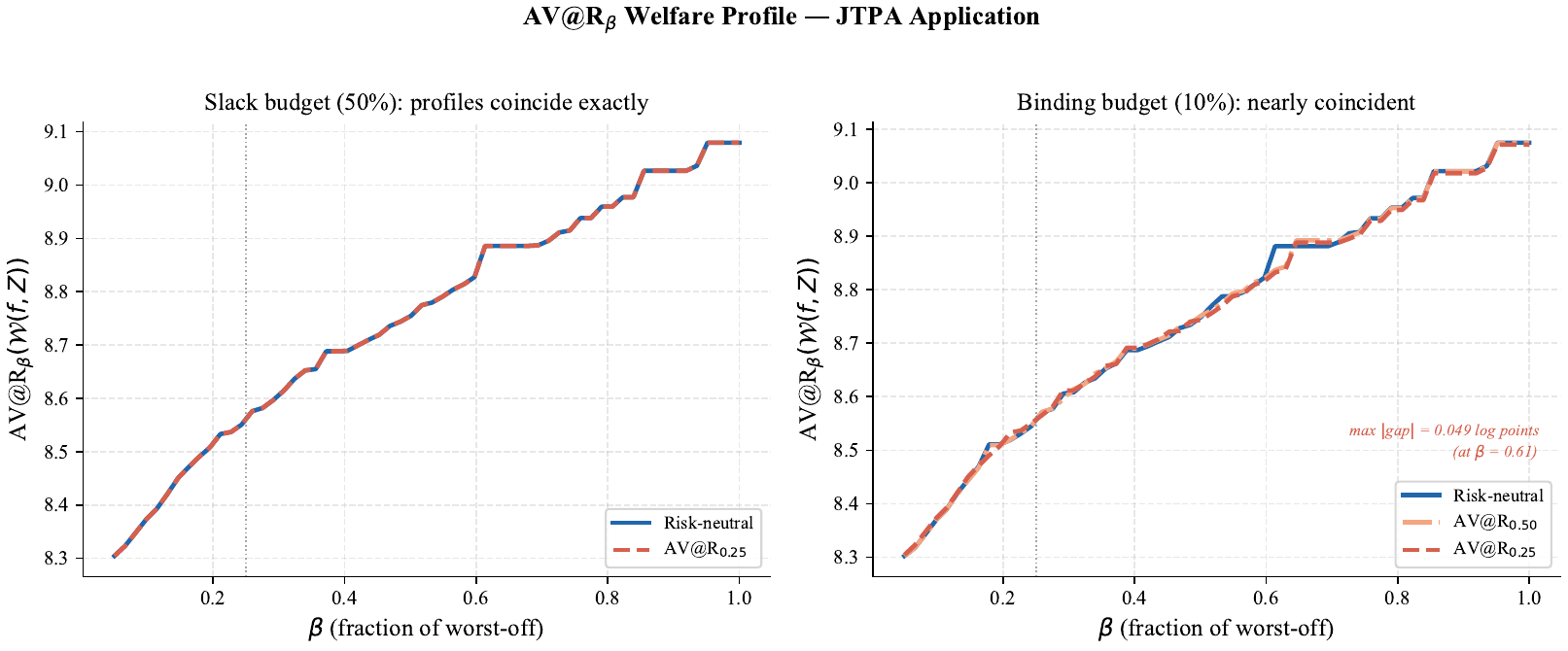}
		\caption{$\mathrm{AV@R}_{\beta}$ welfare profiles on the JTPA data.
			Left: under the slack $50\%$ budget the risk-neutral and risk-averse
			profiles coincide exactly. Right: under the binding $10\%$ budget
			they remain nearly coincident, with a maximal gap below $0.05$ log
			points.}
		\label{Fig_JTPA_AVaR}
	\end{figure}
	
	\subsection{Why: The Gains--Baseline Configuration}\label{app:jtpa_why}
	
	The mechanism paragraph of the empirical discussion in the main text identified
	two ingredients that give a risk-averse criterion bite: treatment
	gains large enough to move individuals across the welfare-tail
	threshold, and a negative association between gains and baseline
	welfare so that the movable individuals are tail individuals. The
	JTPA data lack both, and Table~\ref{Tab_Config} quantifies the
	contrast. On the Card data, $96\%$ of estimated gains are positive
	and the 90th percentile of $|\widehat{\Delta}|$ is five times the
	standard deviation of baseline welfare: treatment is a powerful lever
	on the shape of the welfare distribution. On the JTPA data, positive
	gains are scarce ($22\%$) and typical gain magnitudes are comparable
	to baseline dispersion ($0.43$ against $0.40$), while the
	gains--baseline correlation is weak ($-0.16$ against $-0.33$): the
	composition of the welfare tail is essentially fixed by covariates,
	and no feasible assignment reshapes it. Figure~\ref{Fig_JTPA_Config}
	displays the configuration directly.
	
	\begin{table}[t]
		\centering
		\caption{The gains--baseline configuration: Card vs.\ JTPA.}
		\label{Tab_Config}
		\begin{tabular}{lcc}
			\hline\hline
			& Card (1995) & JTPA \\
			\hline
			Share of positive estimated gains, $1\{\widehat{\Delta}_i > 0\}$
			& $0.960$ & $0.221$ \\
			90th percentile of $|\widehat{\Delta}_i|$ & $1.192$ & $0.432$ \\
			Std.\ dev.\ of baseline welfare $\hat{\mu}_{0,i}$ & $0.239$ & $0.399$ \\
			Correlation of $\widehat{\Delta}_i$ and $\hat{\mu}_{0,i}$
			& $-0.325$ & $-0.160$ \\
			\hline
			$\mathrm{AV@R}_{0.25}$ gain of risk-averse over risk-neutral rule
			& $0.171$ & $0.005$ \\
			\hline\hline
		\end{tabular}
		\par\smallskip
		{\footnotesize\emph{Notes:} Welfare in log wages (Card) and log
			30-month earnings (JTPA). Bottom row: binding-budget comparisons
			($50\%$ for Card, $10\%$ for JTPA).}
	\end{table}
	
	\begin{figure}[t]
		\centering
		\includegraphics[width=\textwidth]{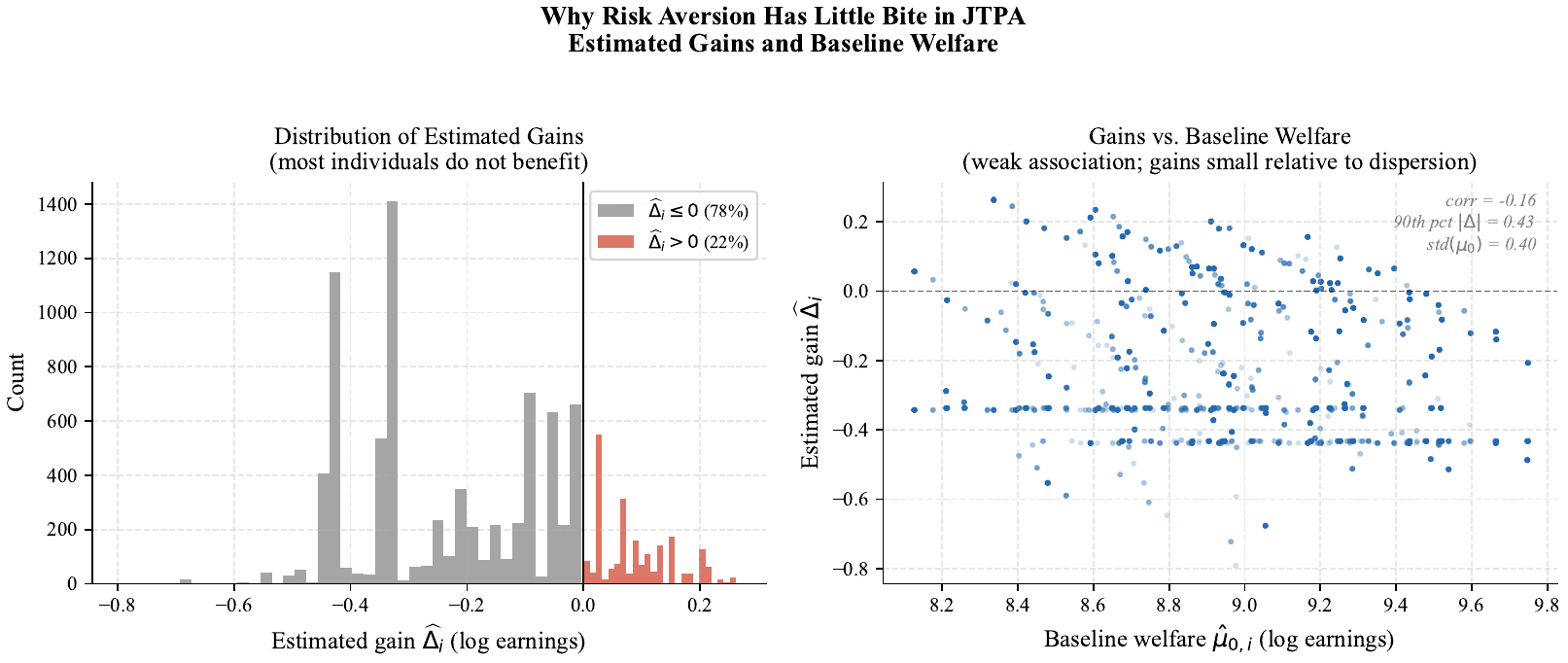}
		\caption{Why risk aversion has little bite on the JTPA data. Left:
			the distribution of estimated gains; most individuals do not benefit
			from treatment. Right: gains against baseline welfare; the
			association is weak and gain magnitudes are small relative to
			baseline dispersion, so treatment cannot reshape the welfare tail.}
		\label{Fig_JTPA_Config}
	\end{figure}
	
	We read these results as a dialogue with \citet{Sasaki_Ura_2024} in the
	most direct sense available: on their own dataset, our framework
	reports that their risk-neutral optimal rule is, to economic
	approximation, also the risk-averse optimal rule. The value added of
	a coherent risk criterion is population-specific, and the statistics
	in Table~\ref{Tab_Config} are computable diagnostics for whether a
	given application is a Card-like or a JTPA-like environment before
	any risk-averse machinery is deployed.
	
	\refstepcounter{section}
	
	\section*{Appendix C: The Mean--Semideviation Risk Measure}\label{app:semidev}
	
	This appendix gives the full, self-contained treatment of the mean--semideviation risk measure ($p=1$) referenced in Section~\ref{s4}: verification of Assumption~\ref{Ass_KusuokaTail}, the local-regularity condition and empirical tail-control lemma this requires, the resulting finite-sample regret bound (Table~\ref{tab:lipschitz}), and a discussion of why the case $p\neq1$ remains open.
	
	\subsection{Verification for $p=1$.} Absolute mean--semideviation, $\rho_{\mathrm{Semi}}(W)=\EE[W]-c\,\mathbb{D}_1^-[W]$ with $c\in[0,1]$, has an explicit Kusuoka representation, obtained by translating Example~2 of \citet[p.~151]{Shapiro_Kusuoka2013} (stated there for a loss random variable $Z=-W$) into our welfare (reward) convention, in which $\mathrm{AV@R}_\alpha^{\text{Shapiro}}(-W)=-\mathrm{AV@R}_{1-\alpha}(W)$ and $\mathrm{AV@R}_1(W)=\EE[W]$ (Example~\ref{Example1}):
	\begin{equation}\label{Semi_Kusuoka_p1}
		\rho_{\mathrm{Semi}}(W)=\inf_{\kappa\in(0,1)}\Bigl\{(1-c\kappa)\,\mathrm{AV@R}_1(W)+c\kappa\,\mathrm{AV@R}_\kappa(W)\Bigr\},
	\end{equation}
	i.e., a two-point Kusuoka measure $\mu_\kappa=(1-c\kappa)\delta_1+c\kappa\delta_\kappa$. Writing $g(\kappa)$ for the bracketed expression, $g(0^+)=g(1)=\EE[W]$ (using $\mathrm{AV@R}_1(W)=\EE[W]$ at both endpoints) while $g(\kappa)\leq\EE[W]$ throughout, with strict inequality whenever $c>0$ and $W$ is non-degenerate; hence, whenever $\rho_{\mathrm{Semi}}(W)<\EE[W]$, the infimum in~(\ref{Semi_Kusuoka_p1}) is attained at an \emph{interior} $\kappa^\star\in(0,1)$, not at the boundary $\kappa\to0$. Translating Shapiro's extreme-point construction ($\kappa=\Pr\{Z>\EE[Z]\}$ for the loss variable) into our convention gives $\kappa^\star(f)=\Pr\{\mathcal{W}(f,Z)<\EE[\mathcal{W}(f,Z)]\}$, the probability that welfare falls below its own mean. Consequently Assumption~\ref{Ass_KusuokaTail} holds for $p=1$ mean--semideviation with
	$$
	\alpha_{\min}=\inf_{f\in\mathcal{F}}\kappa^\star(f)=\inf_{f\in\mathcal{F}}\Pr\{\mathcal{W}(f,Z)<\EE[\mathcal{W}(f,Z)]\},
	$$
	\emph{provided} this infimum is strictly positive---a mild non-degeneracy condition on $\mathcal{F}$ (it fails only if some policy pushes welfare to be almost surely at or above its own mean, i.e., an essentially one-sided welfare distribution) rather than a property that needs to be independently assumed about the risk measure itself.
	
	\subsection{Local Regularity and Empirical Tail Control}
	
	Assumption~\ref{Ass_KusuokaTail} requires the tail-support condition to hold for $\mathcal{W}_n(f,Z)$ as well, i.e., for the \emph{empirical} analogue $\kappa_n^\star(f)\triangleq\tfrac1n\sum_{i=1}^n\mathds{1}\{\mathcal{W}(f,Z_i)<\bar{\mathcal{W}}_n(f)\}$, where $\bar{\mathcal{W}}_n(f)\triangleq\tfrac1n\sum_i\mathcal{W}(f,Z_i)$. The following shows that population non-degeneracy transfers to the sample with high probability, uniformly over $\mathcal{F}$, under one additional mild regularity condition.
	
	\begin{ass}[Local regularity of the welfare distribution]\label{Ass_LocalReg}
		There exist constants $c_0,r_0>0$ such that, writing $\mu(f)\triangleq\EE[\mathcal{W}(f,Z)]$ and $F_f(t)\triangleq\Pr\{\mathcal{W}(f,Z)\leq t\}$, for every $f\in\mathcal{F}$ and all $s,t$ with $|s-\mu(f)|,|t-\mu(f)|\leq r_0$,
		$$
		|F_f(s)-F_f(t)|\leq c_0|s-t|.
		$$
	\end{ass}
	
	This holds, e.g., whenever $\mathcal{W}(f,Z)$ has a density bounded by $c_0$ in a neighborhood of its own mean, uniformly over $f\in\mathcal{F}$; it rules out policies under which welfare has an atom, or an arbitrarily steep CDF, exactly at its mean.
	
	\smallskip
	Let Assumptions~\ref{Assump_MTE}, \ref{Assumption2_OCE_Bounded}, \ref{Ass_Y0_Bounded}, and~\ref{Ass_LocalReg} hold, and assume the class $\Pi$ of admissible treatment rules is a VC class with finite VC dimension $V<\infty$. Suppose $\alpha_{\min}\triangleq\inf_{f\in\mathcal{F}}\kappa^\star(f)>0$. Then for any $\delta\in(0,1]$, with probability at least $1-\delta$,
	\begin{equation}\label{Empirical_Kusuoka_bound}
		\sup_{f\in\mathcal{F}}\bigl|\kappa_n^\star(f)-\kappa^\star(f)\bigr|
		\leq C\sqrt{\frac{\textcolor{black}{8V\log(en/(8V))}}{n}}+c_0\left(2\,\EE[\mathfrak{R}_n(\mathcal{F},Z^n)]+K\sqrt{\frac{2\log(4/\delta)}{n}}\right),
	\end{equation}
	for an absolute constant $C$. Consequently, for every $n$ large enough that the right side of~(\ref{Empirical_Kusuoka_bound}) is less than $\alpha_{\min}/2$, Assumption~\ref{Ass_KusuokaTail} holds for $\mathcal{W}_n(f,Z)$ with tail constant $\alpha_{\min}/2$, uniformly over $f\in\mathcal{F}$, with probability at least $1-\delta$.\label{Lem_Empirical_Kusuoka}

	\smallskip
	\noindent\textit{Proof.} Write $F_{n,f}(t)\triangleq\tfrac1n\sum_i\mathds{1}\{\mathcal{W}(f,Z_i)\leq t\}$ for the empirical CDF and $\hat\mu(f)\triangleq\bar{\mathcal{W}}_n(f)$. Decompose
	$$
	\kappa_n^\star(f)-\kappa^\star(f)=\underbrace{\bigl[F_{n,f}(\hat\mu(f))-F_f(\hat\mu(f))\bigr]}_{(\mathrm{I})}+\underbrace{\bigl[F_f(\hat\mu(f))-F_f(\mu(f))\bigr]}_{(\mathrm{II})}.
	$$
	\textit{Bounding $(\mathrm{I})$ uniformly.} Since $\mathcal{W}(f,z)=\EE[Y_0\mid z]+\pi(z)\int_0^1\mathrm{MTE}(x,\nu)\,d\nu$ with $\pi\in\Pi$ ranging over a VC class of dimension $V$ and the remaining terms fixed (not indexed by $\pi$), the sub-level set $\{z:\mathcal{W}(f,z)\leq t\}$ equals $[\pi^{-1}(0)\cap C_0(t)]\cup[\pi^{-1}(1)\cap C_1(t)]$ for fixed, one-dimensional threshold families $C_0(t),C_1(t)$ of VC dimension $1$. \textcolor{black}{Crucially, for fixed $\pi$ this reduces to a single threshold sweep of the fixed function $h_\pi(z)=\EE(Y_0\mid z)+\pi(z)\int_0^1\mathrm{MTE}(x,\nu)\,d\nu$: any $\pi,\pi'$ inducing the same restriction pattern on $n$ points give identical values of $h_\pi$ there, hence the same achievable threshold patterns, so $\Delta_{\{z:\mathcal{W}(f,z)\leq t\}}(n)\leq\Delta_\Pi(n)\cdot(n+1)\leq(en/V)^V(n+1)$ for every $n\geq V$, the last step by the standard Sauer--Shelah corollary (cf.\ the general Boolean-combination preservation principle in \citealp[Lemma~2.6.17--2.6.18]{VaartWellner1996}). Taking $n=8V$ gives $(en/V)^V(n+1)<2^n$ for every $V\geq1$, so the class $\{z:\mathcal{W}(f,z)\leq t\}_{f\in\mathcal{F},t\in\mathbb{R}}$ has VC dimension at most $8V$.} The standard VC uniform convergence bound (the same tool underlying $\mathfrak{R}_n(\mathcal{F},Z^n)\leq M\sqrt{2V\log(en/V)/n}$, \citealp[Thm.~26.5]{Shalev-Shwartz_Ben-David_2014}) then gives $\sup_{f,t}|F_{n,f}(t)-F_f(t)|\leq C\sqrt{\textcolor{black}{8V\log(en/(8V))}/n}$ with probability at least $1-\delta/2$, for an absolute constant $C$; in particular this bounds $|(\mathrm{I})|$ pointwise in $f$.
	
	\smallskip
	\textit{Bounding $(\mathrm{II})$ uniformly.} By Assumption~\ref{Ass_LocalReg}, whenever $|\hat\mu(f)-\mu(f)|\leq r_0$, $|(\mathrm{II})|\leq c_0|\hat\mu(f)-\mu(f)|$. But $\sup_f|\hat\mu(f)-\mu(f)|=\sup_f|\bar{\mathcal{W}}_n(f)-\EE[\mathcal{W}(f,Z)]|$ is exactly the risk-neutral ($\rho=\EE$) special case of Proposition~\ref{Prop_UC}(i) (equivalently, Corollary~\ref{Cor_specific}(iii) with $L_\varphi=1$), so it is bounded by $2\,\EE[\mathfrak{R}_n(\mathcal{F},Z^n)]+K\sqrt{2\log(4/\delta)}/\sqrt n$ with probability at least $1-\delta/2$; for $n$ large enough that this bound is below $r_0$, the Lipschitz bound on $(\mathrm{II})$ applies uniformly over $f\in\mathcal{F}$.
	
	\smallskip
	Combining the two bounds via a union bound over the two probability-$\delta/2$ events gives~(\ref{Empirical_Kusuoka_bound}). The final claim follows since $\kappa_n^\star(f)\geq\kappa^\star(f)-\sup_f|\kappa_n^\star(f)-\kappa^\star(f)|\geq\alpha_{\min}-\alpha_{\min}/2=\alpha_{\min}/2$ for every $f\in\mathcal{F}$ once the right side of~(\ref{Empirical_Kusuoka_bound}) is below $\alpha_{\min}/2$. \eproof
	
	This closes the remaining gap for $p=1$, giving Assumption~\ref{Ass_KusuokaTail} an explicit, high-probability finite-sample threshold on $n$ rather than leaving it as an unverifiable ``for every $n$'' postulate.
	
	\smallskip
	\subsection{Regret bound.} \textbf{Mean--semideviation, $p=1$} ($c\in[0,1]$, via Proposition~\ref{Prop_Kusuoka_Symmetrization} and~(\ref{Semi_Kusuoka_p1})): provided $\alpha_{\min}\triangleq\inf_{f\in\mathcal{F}}\Pr\{\mathcal{W}(f,Z)<\EE[\mathcal{W}(f,Z)]\}>0$ (Assumption~\ref{Ass_LocalReg} holds) and $n$ exceeds the explicit, data-generating-process-dependent threshold above (so that $\kappa_n^\star(f)\geq\alpha_{\min}/2$ uniformly over $f\in\mathcal{F}$ with probability at least $1-\delta$), with probability at least $1-\delta$,
	$$
	\mathcal{R}_\rho(\hat f_n) \le \frac{4}{\alpha_{\min}} Q_n(\delta)
	+ 4\,\epsilon_n(\alpha_{\min}/2).
	$$
	Here $Q_n(\delta)$ and $\epsilon_n(\cdot)$ are as in
	Theorem~\ref{Thm_Regret} and
	Proposition~\ref{Prop_Kusuoka_Symmetrization}. Substituting the verified
	empirical tail constant $\alpha_{\min}/2$ for $\alpha_{\min}$ throughout
	Theorem~\ref{Thm_Regret}(ii) doubles both terms: the leading coefficient
	becomes $4/\alpha_{\min}$, and the residual becomes
	$4\epsilon_n(\alpha_{\min}/2)=8\epsilon_n(\alpha_{\min})$. The
	corresponding entry in Table~\ref{tab:lipschitz} is therefore
	$2/\alpha_{\min}$, twice the generic Kusuoka constant $1/\alpha_{\min}$.
	
	\smallskip
	\subsection{The open case}The case $p\neq1$ remains open. Shapiro's construction is special to $p=1$: it exploits that the dual set of $\mathbb{D}_1^-$ (an $L^1$-type deviation) is generated by a single extreme point per Kusuoka measure, which in turn relies on comonotonicity---a property specific to the absolute (order-$1$) semideviation and not known to hold for $\mathbb{D}_p^-$, $p\neq1$, in general (Theorem~2 of \citet{Shapiro_Kusuoka2013} ties the clean discrete Kusuoka representation to comonotonicity). Whether $\mathbb{D}_p^-$ for $p\neq1$ admits a Kusuoka measure with support bounded away from $0$ is accordingly a separate question we have not resolved; the regret bound above is therefore stated only for $p=1$.
	\refstepcounter{section}
	
	\section*{Appendix D: Simulation Study}\label{app:simulation}
	
	This appendix discuss a  full simulation study : the data-generating process, the population analysis of the mean--tail conflict, and the finite-sample regret exercise, together with all four accompanying figures. All risk-averse policies below are computed by exact maximization of the stated criterion over the policy class; the Rockafellar--Uryasev representation
	\begin{equation}
		\label{eq:ru}
		\mathrm{AV@R}_\beta\bigl(\mathcal{W}(f,Z)\bigr)
		\;=\; \max_{c \in \mathbb{R}}
		\Bigl\{ c - \tfrac{1}{\beta}\,
		\mathbb{E}\bigl[(c - \mathcal{W}(f,Z))^{+}\bigr] \Bigr\}
	\end{equation}
	reduces this maximization, for each candidate policy, to a
	one-dimensional concave problem, and the finite policy class introduced
	below permits direct enumeration. No approximation to the risk measure
	is involved at any step.
	
	\subsection{Design}
	\label{subsec:sim-design}
	
	The covariate $X \sim U(0,1)$ is a socioeconomic index, and
	$Z_0 \sim N(0,1)$ is an excluded instrument. Selection follows the
	threshold-crossing model of Section~\ref{s2}, with
	$U \sim U(0,1)$ independent of $(X, Z_0)$,
	\[
	D \;=\; 1\{ U \le p(X, Z_0) \}, \qquad
	p(x, z_0) \;=\; \bigl[\, 0.30 + 0.30\,x + 0.25\,z_0 \,\bigr]_{0.05}^{0.95},
	\]
	where $[\,\cdot\,]_{a}^{b}$ denotes clipping to $[a,b]$; the clipping
	bounds are reported for replicability. Potential outcomes are
	\[
	Y_0 \;=\; \mu_0(X) + \eta, \qquad
	Y_1 \;=\; Y_0 + \delta(X) + \gamma_a - \gamma_b\, U,
	\qquad \eta \sim N(0,\, 0.8^2),
	\]
	with $\mu_0(x) = 1 + 2x$, $\delta(x) = -0.2 + 2.2x$, $\gamma_a = 1$,
	and $\gamma_b = 2$, so that
	$\mathrm{MTE}(x, u) = \delta(x) + \gamma_a - \gamma_b u$ exhibits
	essential heterogeneity: agents with low $U$ have both higher
	unobserved gains and higher treatment take-up. The integrated MTE is
	\[
	\bar{\Delta}(x) \;=\; \int_0^1 \mathrm{MTE}(x,u)\, du
	\;=\; \delta(x) + \gamma_a - \tfrac{\gamma_b}{2} \;=\; \delta(x),
	\]
	so conditional welfare under policy $\pi$ is
	$\mathcal{W}(\pi, Z) = \mu_0(X) + \pi(Z)\,\bar{\Delta}(X)$. Two features of
	the design do deliberate work. First, both the baseline $\mu_0$ and the
	gain $\bar{\Delta}$ increase in $X$: the individuals who benefit most
	from treatment are also the best off without it, so a budget-constrained
	planner faces a genuine conflict between the mean and the tail.
	Second, $\bar{\Delta}(x) < 0$ for $x < x_{\dagger} \equiv 0.091$:
	treatment \emph{harms} the most disadvantaged individuals. This zone
	disciplines the comparison across planners in a way that a uniformly
	beneficial treatment cannot.
	
	Note that the outcome noise $\eta$ does not enter
	$\mathcal{W}(f,Z)$, which is a conditional expectation given $Z$; risk
	aversion in our framework operates on the cross-sectional distribution
	of conditional welfare, not on idiosyncratic outcome variance. The
	noise matters only for the difficulty of \emph{estimating} the MTE
	ingredients, which is the subject of Section~\ref{subsec:sim-regret}.
	
	The policy class consists of interval rules on the index,
	\[
	\Pi \;=\; \Bigl\{ \pi_{t,w}(x) = 1\{ t \le x \le t + w \}
	\;:\; 0 \le t,\; t + w \le 1,\; 0 \le w \le \bar{w} \Bigr\},
	\qquad \bar{w} = 0.5,
	\]
	so the budget constraint caps the treated share at 50\%, mirroring the
	empirical application of the main text. The class has
	finite VC dimension, and we compute all population optima by
	enumeration on a grid over $(t, w)$ with the population criterion
	evaluated exactly.
	
	\subsection{Population optima: the mean--tail conflict}
	\label{subsec:sim-population}
	
	Figure~\ref{fig:sim-rules} displays the primitives and the optimal
	intervals. The risk-neutral planner solves
	$\max_{f \in \mathcal{F}} \mathbb{E}[\mathcal{W}(f,Z)]$ and, because
	$\bar{\Delta}$ is increasing, exhausts the budget on the top half,
	treating $[0.50,\, 1.00]$. The AV@R planners shift the interval toward
	the bottom of the distribution: the $\mathrm{AV@R}_{0.50}$ optimum is
	$[0.17,\, 0.67]$ and the $\mathrm{AV@R}_{0.25}$ optimum is
	$[0.09,\, 0.49]$. The lower endpoint of the latter coincides with
	$x_{\dagger}$: the tail-focused planner extends treatment as far down
	the distribution as it helps, and stops precisely where gains turn
	negative. We emphasize that this boundary behavior is not imposed; it
	is found by the optimizer, and it illustrates that risk aversion in
	the planner's criterion is not a mandate to treat the worst-off
	regardless of consequences, but to raise the lower tail of welfare by
	whatever assignment achieves it.
	
	The naive planner estimates gains ignoring the instrument, comparing
	treated and untreated outcomes conditional on $X$ alone. Under
	essential heterogeneity this overstates gains, and the overstatement
	is largest where take-up is lowest --- here, at the bottom of the
	index. In population, the naive gain estimate is
	$\tilde{\Delta}(x) = \delta(x) + \gamma_a - \gamma_b\,\bar{p}(x)/2 > \bar{\Delta}(x)$,
	where $\bar{p}(x) = \mathbb{E}[p(x, Z_0)]$, and $\tilde{\Delta}$ is
	strictly positive on all of $[0,1]$: the naive planner does not see
	the negative-gain zone at all. Its $\mathrm{AV@R}_{0.25}$ rule treats
	$[0.00,\, 0.50]$, pushing treatment onto individuals whom treatment
	harms. Figure~\ref{fig:sim-cdfs} shows the consequence in the welfare
	distribution: at the first percentile, welfare under the naive rule is
	$0.84$, against $1.02$ under either the risk-neutral or the correctly
	computed risk-averse rule. The naive planner, attempting to protect
	the worst-off, is the only one who makes them worse off. The true
	$\mathrm{AV@R}_{0.25}$ value of the naive rule is $1.325$, against the
	oracle value $1.361$; this gap reappears as the asymptote of the naive
	regret curve below.
	
	Figure~\ref{fig:sim-profile} traces
	$\mathrm{AV@R}_\beta(\mathcal{W}(f,Z))$ as a function of $\beta$ for
	each fixed rule. The risk-averse rules dominate the risk-neutral rule
	for tail levels $\beta$ roughly below $0.64$ and are dominated above,
	with the crossing reflecting the price of tail protection paid in mean
	welfare ($2.18$ under the $\mathrm{AV@R}_{0.25}$ rule against $2.73$
	under the risk-neutral rule). At the extreme tail
	($\beta \lesssim 0.09$) the risk-neutral and risk-averse profiles
	coincide: the individuals below $x_{\dagger}$ cannot be helped by this
	treatment, so no assignment separates the rules there. The profile
	thus makes visible both what risk aversion buys and what no policy in
	the class can buy.
	
	\subsection{Regret of the empirically selected policy}
	\label{subsec:sim-regret}
	
	We now turn to the finite-sample exercise that
	Theorem~\ref{Thm_Regret} governs. For each sample size
	$n \in \{250, 500, 1000, 2500, 5000, 10000\}$ and each of $120$
	Monte Carlo replications, we estimate the MTE ingredients from the
	sample --- a linear propensity fit for $\hat{p}$, followed by the
	regression of $Y$ on $\bigl(1, X, \hat{p}, \hat{p}X, \hat{p}^2\bigr)$
	implied by the local instrumental variables representation, from which
	$\hat{\mu}_0(\cdot)$ and $\hat{\bar{\Delta}}(\cdot)$ are recovered ---
	and select $\hat{\pi}_n \in \Pi$ by maximizing the \emph{empirical}
	criterion (the sample mean, or the sample $\mathrm{AV@R}_\beta$ via
	the empirical analogue of \eqref{eq:ru}) over the class. The regret we
	report is the population one:
	$\rho(\mathcal{W}(f^{*},Z)) - \rho(\mathcal{W}(\hat{f}_n,Z))$, with
	both terms evaluated at the true distribution. Nothing about the
	selection step uses knowledge of the data-generating process.
	
	Figure~\ref{fig:sim-regret} reports mean regret on logarithmic axes.
	Three features bear emphasis. First, regret vanishes for all
	MTE-based criteria, and at each sample size the curves are ordered
	exactly as the Lipschitz constants of
	Proposition~\ref{Prop_UC} predict: $L = 1$ for the mean and
	$L = 1/\beta$ for $\mathrm{AV@R}_\beta$, so the $\beta = 0.25$
	criterion pays roughly the constant-factor premium over the
	risk-neutral criterion that the theory anticipates --- a
	multiplicative cost in sample complexity, not a qualitative barrier.
	Second, the observed decay is faster than the $n^{-1/2}$ reference
	line. This is expected: Theorem~\ref{Thm_Regret} is a uniform upper
	bound, and in a finite class with a strict margin between the optimum
	and its nearest competitor, regret decays faster once the optimal
	rule is identified with high probability. The figure should therefore
	be read as confirming the ordering and the sufficiency of the rate,
	not as estimating it. Third, the regret of the naive planner does not
	vanish: it converges to the population gap of $0.036$ computed above.
	Endogenous selection thus enters the figure as an irreducible bias
	that no sample size repairs --- the graphical counterpart of the
	observation, made in the Introduction, that approaches built
	on unconfoundedness are not applicable in this environment.
	
	\begin{figure}[t]
		\centering
		\includegraphics[width=0.85\textwidth]{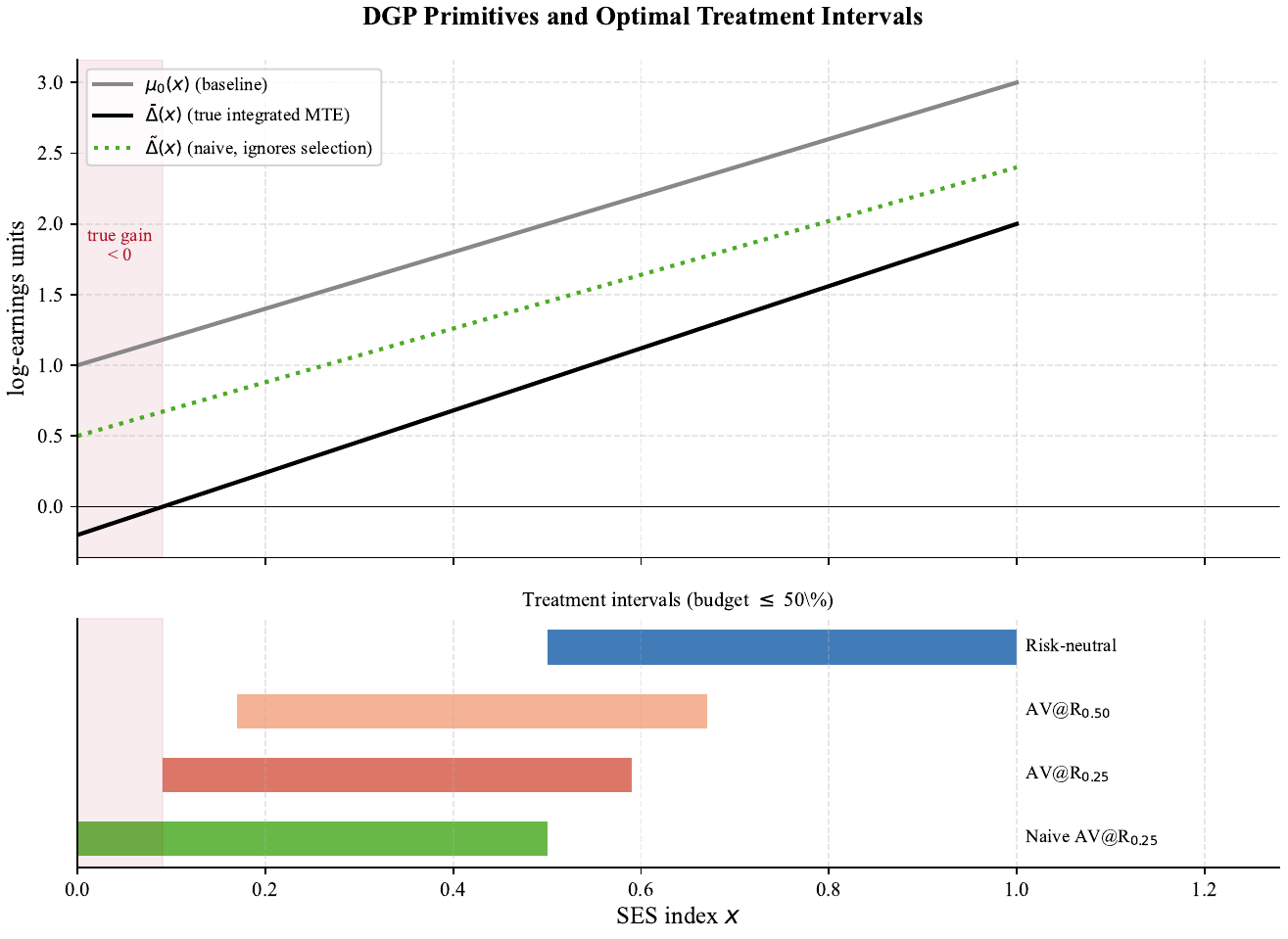}
		\caption{\textbf{DGP primitives and optimal treatment intervals.}
			Top panel: baseline welfare $\mu_0(x)$, true integrated MTE
			$\bar{\Delta}(x)$, and the naive planner's population gain estimate
			$\tilde{\Delta}(x)$; the shaded region marks the negative-gain zone
			$x < x_{\dagger} = 0.091$. Bottom panel: optimal intervals under
			each criterion (budget $\le 50\%$). The $\mathrm{AV@R}_{0.25}$
			interval begins at $x_{\dagger}$; the naive interval extends into
			the negative-gain zone.}
		\label{fig:sim-rules}
	\end{figure}
	
	\begin{figure}[t]
		\centering
		\includegraphics[width=0.95\textwidth]{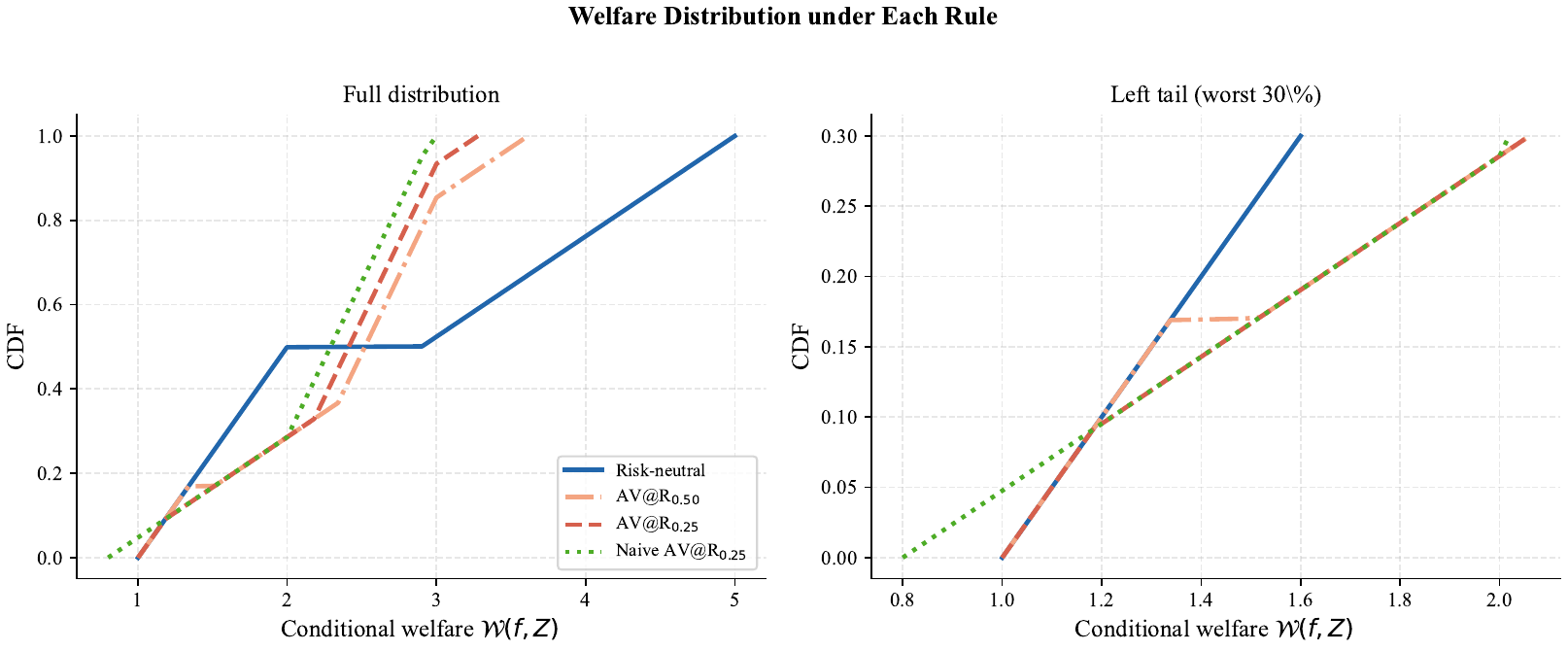}
		\caption{\textbf{Welfare distributions under each rule.} Left: full
			CDFs. Right: lower 30\%. The risk-averse rules first-order dominate
			the risk-neutral rule over the tail region above $x_{\dagger}$; the
			naive rule is dominated by all others at the extreme lower tail,
			where it assigns harmful treatment.}
		\label{fig:sim-cdfs}
	\end{figure}
	
	\begin{figure}[t]
		\centering
		\includegraphics[width=0.75\textwidth]{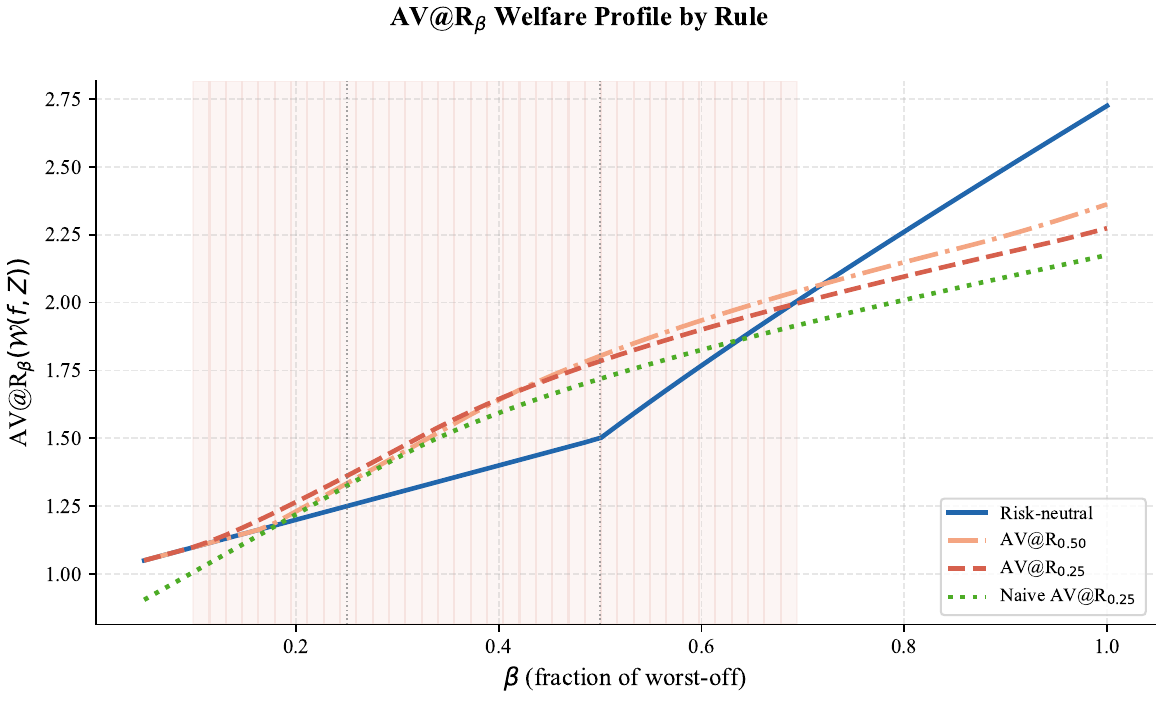}
		\caption{\textbf{$\mathrm{AV@R}_\beta$ welfare profiles.} Each curve
			fixes a rule and varies the evaluation level $\beta$. Shading marks
			the region where the $\mathrm{AV@R}_{0.25}$ rule dominates the
			risk-neutral rule. Profiles coincide below
			$\beta \approx 0.09$, where the negative-gain zone renders the tail
			unimprovable within the class.}
		\label{fig:sim-profile}
	\end{figure}
	
	\begin{figure}[t]
		\centering
		\includegraphics[width=0.75\textwidth]{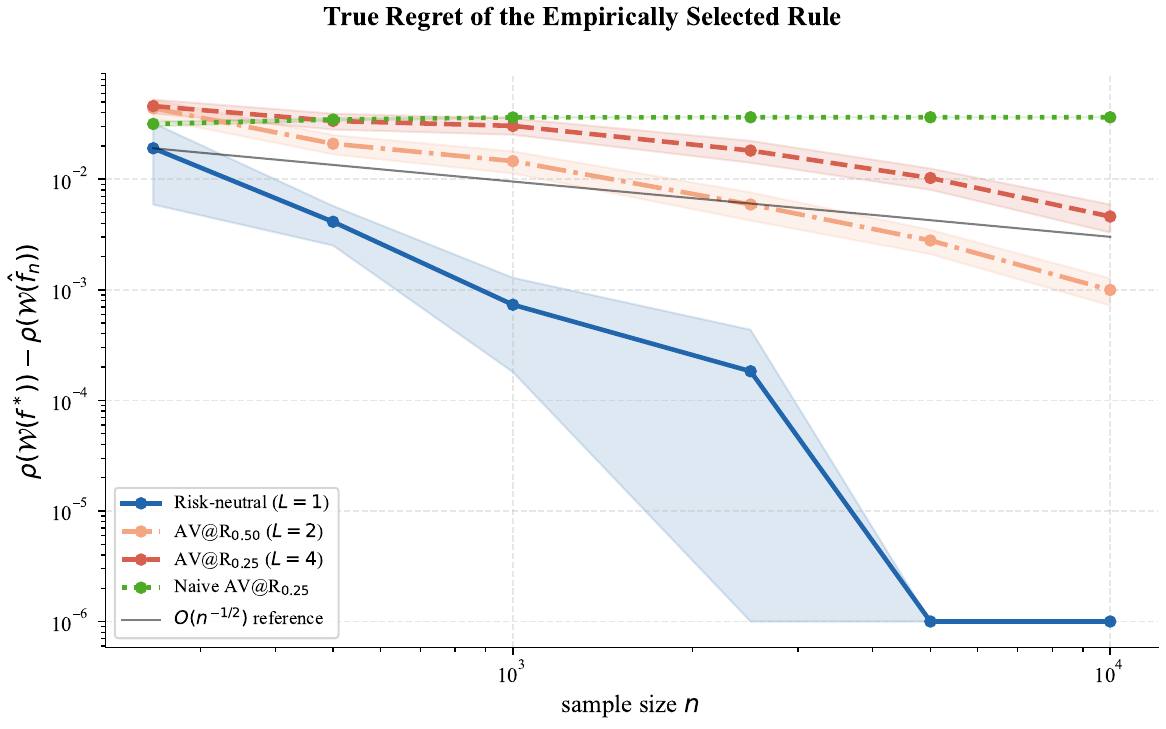}
		\caption{\textbf{True regret of the empirically selected rule.}
			Mean population regret over $120$ replications (log--log scale),
			with $95\%$ bands and an $n^{-1/2}$ reference. MTE-based criteria
			converge with constants ordered by the Lipschitz constant
			($L = 1, 2, 4$); the naive rule's regret converges to its
			population bias of $0.036$ rather than to zero.}
		\label{fig:sim-regret}
	\end{figure}

	\bibliographystyle{plainnat}
	\bibliography{Ref_MTE}
	
\end{document}